\documentclass[journal]{IEEEtran}
\IEEEoverridecommandlockouts
\makeatletter
\def\ps@headings{%
\def\@oddhead{\mbox{}\scriptsize\rightmark \hfil \thepage}%
\def\@evenhead{\scriptsize\thepage \hfil \leftmark\mbox{}}%
\def\@oddfoot{}%
\def\@evenfoot{}}
\makeatother
\usepackage{subfig}
\usepackage{colortbl}
\usepackage{bm}

\usepackage{graphicx}  
\usepackage{url}       

\usepackage{amsmath}   
\usepackage{cite}

\usepackage{amsfonts,amssymb}

\usepackage{stfloats}

\usepackage{cases}
\usepackage{algorithm}
\usepackage{multirow}
\usepackage{algorithmic}
\usepackage{stfloats}
\usepackage{epstopdf}
\usepackage{esint}
\usepackage{mathtools}

\usepackage{xcolor}
\usepackage{xpatch}
\makeatletter
\def\changeBibColor#1{%
	\in@{#1}{}
	\ifin@\else\normalcolor\fi
}
\xpatchcmd\@bibitem
{\item}
{\changeBibColor{#1}\item}
{}{\fail}
\xpatchcmd\@lbibitem
{\item}
{\changeBibColor{#2}\item}
{}{\fail}

\newcommand{\Rmnum}[1]{\expandafter\@slowromancap\romannumeral #1@}

\makeatother

\newtheorem{proposition}{{Proposition}}

\newtheorem{remark}{{Remark}}

\newcommand{\ls}[1]
    {\dimen0=\fontdimen6\the\font
     \lineskip=#1\dimen0
     \advance\lineskip.5\fontdimen5\the\font
     \advance\lineskip-\dimen0
     \lineskiplimit=.9\lineskip
     \baselineskip=\lineskip
     \advance\baselineskip\dimen0
     \normallineskip\lineskip
     \normallineskiplimit\lineskiplimit
     \normalbaselineskip\baselineskip
     \ignorespaces
    }

\begin{document}

\title{Unified Analytical Framework for Emergency RIS-UAV Networks Under Practical Impairments}
\vspace{10pt}

\author{
\IEEEauthorblockN{Yinong Chen, \emph{Graduate Student Member}, \emph{IEEE}, 
Wenchi Cheng, \emph{Senior Member}, \emph{IEEE}, Jingqing Wang, \emph{Member}, \emph{IEEE}, 
and Jiangzhou Wang, \emph{Fellow}, \emph{IEEE}}\vspace{-10pt}


\thanks{
This work was supported in part by the National Key R\&D Program of China under Grant 2024YFC3016000.

Yinong Chen, Wenchi Cheng, and Jingqing Wang are with the
State Key Laboratory of Integrated Services Networks, Xidian University, Xi'an, 710071, China (e-mails:  yinongchen@stu.xidian.edu.cn; wccheng@xidian.edu.cn; jqwangxd@xidian.edu.cn).

Jiangzhou Wang is with the School of Information Science and Engineering, Southeast University, and Purple Mountain Laboratories, Nanjing 211119, China (e-mail: j.z.wang@kent.ac.uk).
}
}

\maketitle

\begin{abstract}
Heterogeneous unmanned aerial vehicle (UAV) networks embedded with reconfigurable intelligent surfaces (RISs) present a promising paradigm for emergency wireless communications (EWC), offering enhanced coverage and resilience in harsh environments. However, extreme conditions in disaster areas necessitate robust performance evaluation under practical impairments, including outdated/imperfect channel state information (CSI) and discrete RIS phase shifts. Existing works lack a unified analytical framework for modeling CSI errors, employing inconsistent approaches that treat errors either as channel gain or as equivalent interference, leading to ambiguous benchmarks. To address this, we propose the $\zeta$-Model, a unified receiver-equivalent signal-to-noise (SNR) framework that continuously parameterizes residual-error exploitability via $\zeta$. This framework unifies the information-theoretic model (ITM) and the engineering baseline model (EBM) as the optimistic and pessimistic benchmark receiver treatments, while incorporating the simplified engineering model (SEM) as a tractable approximation. By employing the Fisher-Snedecor $\mathcal{F}$ distribution to capture severe fading and shadowing, we derive moment-matching-based closed-form or finite-sum approximate expressions and asymptotic expressions for average capacity (AC), effective capacity (EC), and outage probability (OP) under the proposed unified framework and its boundary cases. Validated by Monte Carlo simulations, our framework quantifies performance limits and provides crucial insights for designing robust and efficient EWC systems under various channel conditions and system impairments.

\end{abstract}

\vspace{5pt}


\section{Introduction}
\IEEEPARstart{N}{atural} disasters, such as earthquakes and hurricanes, pose severe threats to public safety by devastating terrestrial infrastructures and disrupting communication links. Therefore, rapidly reestablishing connectivity is crucial for effective coordination and rescue operations\cite{EWC2,EWCsurvey}. 
In post-disaster areas, collapsed buildings and debris create a complex propagation environment with severe multi-path fading and shadowing, making  the traditional approach of employing ground vehicles as temporary base stations (BSs) unable to meet the dynamic and reliable information acquisition requirements of emergency wireless communication (EWC)\cite{Yao}.  

In recent years, towards the advances in miniaturization, low cost, and autonomous control, unmanned aerial vehicles (UAVs) have been regarded as efficient techniques for establishing EWC networks\cite{UAVNetwork}. Generally, high-altitude UAVs serve as aerial BSs for wide coverage due to their endurance and payload capacity, while low-altitude rotary-wing UAVs operate as relays for specific communication demands through cooperative scheduling\cite{UAVsurvey}. Nevertheless, the limited battery life of the low-altitude relays, constrained by on-board RF circuits, presents a significant bottleneck for the overall networks. 
To overcome this issue,  reconfigurable intelligent surface (RIS)-equipped UAVs (R-UAVs) have emerged as a promising and energy-efficient alternative\cite{UAV_RIS3}. 
By manipulating incident signals via programmable phase shifts of RIS elements, R-UAVs can obviate the need for power-intensive RF chains, thereby enhancing network sustainability in emergency scenarios \cite{cao-RIS,cheng2026,UAV_RIS1,UAV_RIS2}. 
For example, the authors in \cite{UAV_RIS1} maximized throughput by jointly optimizing the UAV trajectory, RIS passive beamforming, and source power allocation for each time slot. The authors in \cite{UAV_RIS2} proposed an R-UAV-assisted approach to maximize the secrecy rate while ensuring secure quality of service (QoS)-aware communications. 

To ensure the resilience and stability of the system, the heterogeneous architecture, comprising multiple UAV types as different layers, is preferable in a cooperative manner to provide seamless connectivity in disaster areas. For example, the authors in \cite{heterogeneous_UAV1} proposed a joint optimization framework to optimize overall energy efficiency, subchannel allocation, transmit power, and multi-UAV trajectories for the multi-cell UAV heterogeneous network and sidelink with QoS constraint. The authors in \cite{heterogeneous_UAV2} investigated the cache-enabled UAV emergency communication networks using stochastic geometry approaches. Such multi-tier heterogeneous UAV-assisted network architectures have been recognized as effective paradigms for ensuring resilient and scalable EWC\cite{multiUAV}. 

However, the practical deployment of the above RIS-assisted aerial systems is severely constrained by unavoidable channel uncertainties and hardware imperfections. First, the performance of RIS-aided transmission critically relies on accurate phase configuration, which in turn requires reliable channel state information (CSI). Yet, due to the passive architecture of RISs, cascaded channel estimation incurs considerable training and signaling overhead, inevitably leading to CSI estimation errors \cite{RISperfectCSI,RISestimate}. Second, because of the mobility of aerial BSs as well as processing and transmission delays, the CSI available at the BS may become outdated before data transmission, which can cause a significant performance loss \cite{outdatedCSI1,outdatedCSI2}. Third, continuously tunable phase control is generally impractical at RISs, since their digital control circuits only support discrete phase shifts, thereby introducing residual phase quantization errors \cite{discrete3}. 
In addition, deploying R-UAVs in emergency scenarios faces practical constraints such as payload/weight limitations, onboard power consumption, flight endurance, RIS control overhead, and UAV coordination signaling\cite{EWCsurvey}. These deployment-level constraints further motivate the need for analytical performance expressions that can serve as efficient evaluators in future optimization problems involving power, reliability, RIS configuration, and UAV placement.

Several works have investigated system performance under these practical constraints. For discrete phase quantization, the authors in \cite{discrete3} provided active-RIS-assisted NOMA-ISAC analysis, which modeled the RIS phase error as a uniformly distributed mismatch and characterized the resulting attenuation loss, while the authors in \cite{discretephase} likewise adopted uniform quantization for discrete phase shift for multi-RIS statistical characterization. 
For imperfect and outdated CSI, the authors in \cite{model2s-1} analyzed the average capacity (AC) over Rician fading under outdated CSI, and demonstrated the distinct scenarios where distributed RIS-aided deployment and centralized RIS-aided deployment respectively excel in performance, while the authors in \cite{model2s-2} investigated the impact of outdated CSI on the effective capacity (EC) of RIS-enabled mmWave systems and optimized the RIS and protocol configuration to reduce the signaling overhead. The authors in \cite{model1} analyzed the BER, throughput, and diversity order of a multi-RIS LoRa network over Nakagami-$m$ fading with both imperfect and outdated CSI. 
To further consider dual limitations, 
the authors in \cite{imperfectmodel1} designed robust binary-state selection strategies to maximize the worst-case energy efficiency (EE) under imperfect CSI for both continuous and discrete phase settings, whereas the authors in \cite{imperfectoutmodel1} considered multi-RIS selection strategy and derived the average coverage probability and BER under outdated and imperfect CSI with line-of-sight (LoS) channels. The authors in \cite{Eq1imperfect} studied the outage probability (OP) and AC of active RIS-assisted MISO systems under both imperfect CSI and discrete phase quantization over Nakagami-$m$ fading.


Despite these contributions, existing studies suffer from several limitations. First, the modeling of CSI imperfection is inconsistent: some works incorporate estimation errors in the numerator of the signal-to-noise (SNR) model\cite{imperfectmodel1,imperfectoutmodel1,model1}, while others place them in the denominator as equivalent interference\cite{Eq1imperfect,model2s-1,model2s-2}. 
To date, the side-by-side comparative analysis and inherent mechanism between these two models remain unresolved. Hence, it remains unclear which model is more appropriate for evaluating specific performance metrics, how the performance gap between these two benchmarks behaves under different channel conditions, and what insights this gap provides into the true cost of CSI uncertainty. This leads to the absence of a unified benchmark for quantifying the impact of imperfect and outdated CSI on different performance metrics. 
Second, channels will experience severe multi-path fading and shadowing under complex and dynamic EWC environments, necessitating an accurate and adaptable channel model to indicate reliable system performance. Existing studies predominantly rely on conventional fading models (e.g., Rayleigh, Rician, and Nakagami-$m$), which fail to accurately capture the composite fading dynamics in post-disaster environments. Additionally, highly flexible generalized models, such as the $\lambda$-$\kappa$-$\mu$ and $\alpha$-$\eta$-$\mu$/inverse Gamma distributions\cite{compositechannel,IGA}, can characterize such composite channels. However, their excessive parameters and mathematically complex expressions often lead to prohibitive analytical complexity.

Motivated by the above observations, this paper proposes a unified analytical framework, the $\zeta$-Model. Far from being a mere conceptual categorization, the proposed framework recasts the aggregate CSI residual into two parts: the 
exploitable component and the unexplained component, parameterized by coefficient $\zeta$ to quantify, in a receiver-equivalent sense, the fraction of residual uncertainty exploitable through side information. 
In particular, we apply this framework to a heterogeneous UAV architecture comprising a high-altitude UAV-BS and multiple low-altitude R-UAVs in a post-disaster area. We adopt the modified Fisher-Snedecor $\mathcal{F}$ distribution for air-to-ground links, which offers a favorable balance between empirical suitability and analytical tractability for composite multipath-shadowing channels\cite{tcomF}. Concurrently, practical impairments are incorporated, including imperfect CSI due to estimation error, outdated CSI due to mobility, and discrete RIS phase shifts. Within this unified family, the proposed $\zeta$-Model reveals the existing Information-Theoretic Model (ITM) and Engineering Baseline Model (EBM) as the optimistic full-exploitation benchmark ($\zeta=1$) and pessimistic no-exploitation treatment ($\zeta=0$). Additionally, the Simplified Engineering Model (SEM) is regarded as a tractable approximation of the EBM. 
Based on this, we derive the moment-matching-based approximate distributions of the proposed models and conduct a statistical validation. By employing the Kullback-Leibler (KL) divergence for global behavior and the Mean Relative Error (MRE) for tail accuracy, we not only quantify the performance deviation between the approximate theoretical distributions and Monte Carlo empirical distributions but also demonstrate the operational boundary of SEM failure.  
The closed-form or finite-sum approximate expressions for the AC, OP, and EC are provided, and the asymptotic behaviors are analyzed to gain insights into the performance limits. 
Monte Carlo simulations are conducted to verify the consistency between analytical results and simulations. 

The remainder of this paper is organized as follows: Section II introduces the emergency heterogeneous UAV network, where a unified analytical framework in the presence of practical constraints is proposed, including outdated and imperfect CSI as well as discrete phase shift. Section III presents the statistical analysis of the SNR for the proposed model, followed by the KL divergence and MRE verification. Section IV provides the analytical expressions for the AC, OP, and EC. Simulations are presented in Section V. Finally, Section VI concludes the paper.  

\emph{Notations:} $\left\lVert\cdot \right\rVert $ denotes the Euclidean norm of the argument;
diag$(\cdot)$ denotes a diagonal matrix;
$\angle$ denotes the phase of a complex number; 
$(\cdot)^\ast$ denotes the complex conjugate operation;
$\operatorname{Re}(\cdot)$ and $\operatorname{Im}(\cdot)$ denote the real and imaginary parts, respectively; $\mathbb{E}_{_X}[\cdot]$ denotes the expectation of the RV $X$; $\mu_X$ and $\sigma^2_X$ denote the mean and variance of a RV $X$, respectively; $\mathcal{U}(\cdot, \cdot)$, $\mathcal{N}(\cdot, \cdot)$, $\mathcal{CN}(\cdot, \cdot)$, $\mathrm{Gamma}(\cdot,\cdot$), and $\mathrm{Beta}(\cdot,\cdot)$ denote the uniform, normal, circularly symmetric complex Gaussian (CSCG), Gamma, and Beta distributions, respectively.

\section{System Model}
\begin{figure}
	\centering
	{\includegraphics[width=0.85\columnwidth]{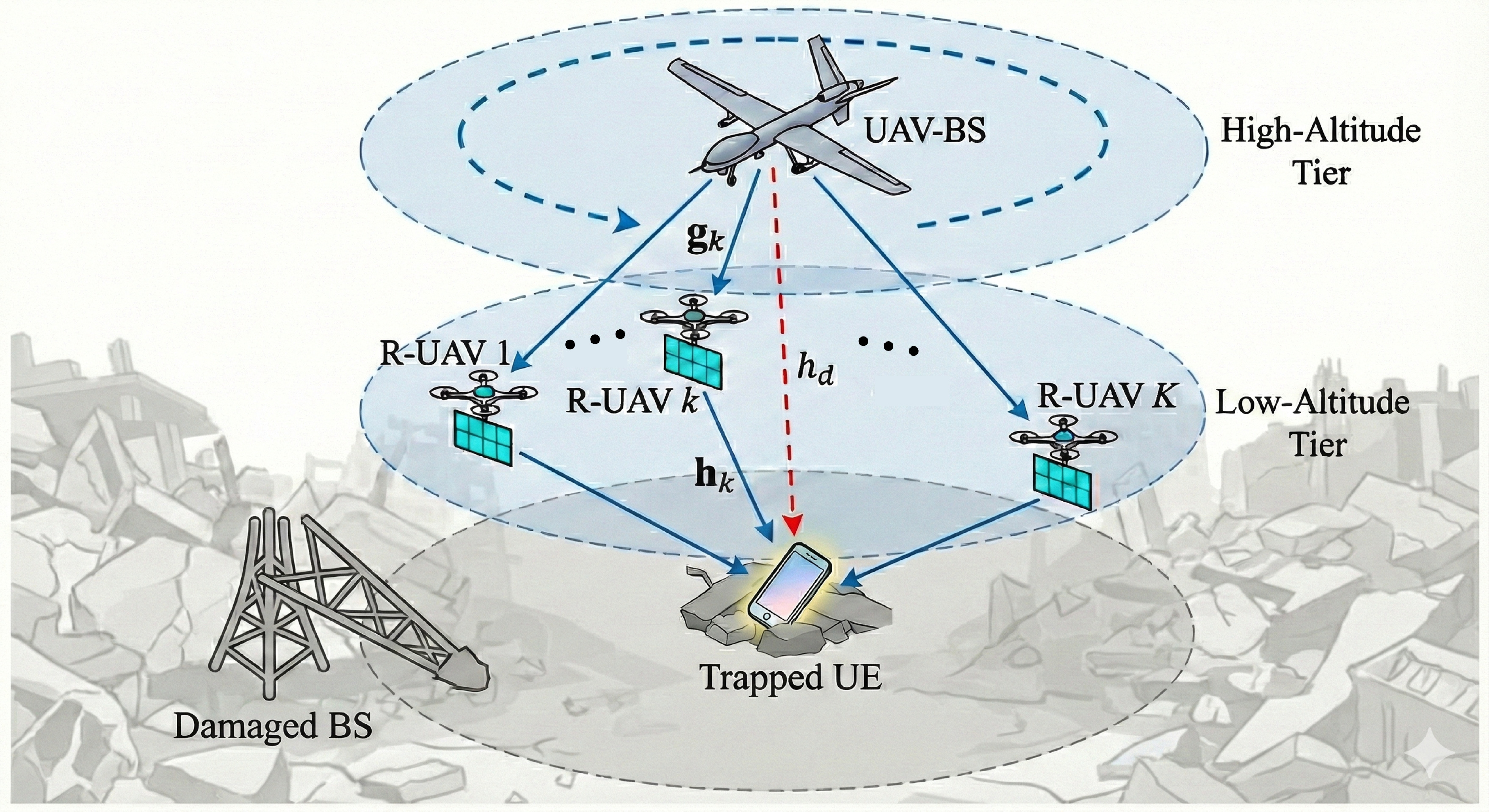}}
	\vspace{-0.2cm}\caption{The emergency heterogeneous UAV system model.}\vspace{-15pt}
	\label{systemmodel}
\end{figure}

Figure \ref{systemmodel} illustrates an emergency heterogeneous UAV system, where terrestrial BSs are damaged, and ground users are irregularly distributed or trapped beneath ruins. To restore communication services, a high-altitude large-scale UAV serves as the UAV-BS located at $\mathbf{q_b}=(x_b,y_b,z_b)$, and $K$ low-altitude small-scale R-UAVs are positioned at $\mathbf{q_k}=(x_k,y_k,z_k)$, respectively. 
The UAV-BS with a single transmit antenna serves the user equipment (UE) located at $\mathbf{q_0}=(x_0,y_0,z_0)$ via direct and multiple reflected links through the R-UAVs
\footnote{Although the direct UAV-BS--UE link is included in the system model, it may still be unreliable in post-disaster environments because of high-altitude distance loss, unfavorable elevation geometry, blockage, and severe shadowing. The R-UAV-assisted cascaded link is therefore introduced as a complementary propagation route with additional spatial flexibility.\vspace{0pt}}.
Let $h_d$ denote the channel coefficient of the direct link from the UAV-BS to the UE. It is assumed that each R-UAV has $N$ passive reflecting elements, hence the channel vectors from the UAV-BS to the $k$-th R-UAV and from the $k$-th R-UAV to the UE are denoted as $\mathbf{g_k}=[g_{k,1}, ..., g_{k,N}]\in \mathbb{C}^{1\times N}$ and $\mathbf{h_k}=[h_{k,1}, ..., h_{k,N}]\in \mathbb{C}^{1\times N}$, respectively. 
The RIS phase shift matrix of the $k$-th {R-UAV} is defined as $\bm{\Phi}_k=\text{diag}(e^{-j\Phi_{k,1}},...,e^{-j\Phi_{k,N}})\in\mathbb{C}^{N\times N}$, where $\Phi_{k,n}\in [0, 2\pi)$ represents the phase shift of the $n$-th reflecting element. For simplicity, the cascaded channel between the UAV-BS and the UE via the $k$-th R-UAV is denoted as $\bm{\chi}_k=[\chi_{k,1},...,\chi_{k,N}]\in\mathbb{C}^{1\times N}$, where each element is given by $\chi_{k,n}=g_{k,n}h_{k,n}e^{-j\Phi_{k,n}}=|g_{k,n}||h_{k,n}|e^{j\angle\chi_{k,n}}$ with $\angle\chi_{k,n}=\angle{g_{k,n}}+\angle{h}_{k,n}-\angle\Phi_{k,n}$. 
Therefore, the received signal at UE can be expressed as follows:\vspace{-3pt}
\begin{equation}
y=\sqrt{P_s}\left(
\frac{h_d}{\sqrt{PL_d}}
+
\sum_{k=1}^{K}\sum_{n=1}^{N}
\frac{\chi_{k,n}}{\sqrt{PL_{k,n}}}\right)x+n,\label{signal}\vspace{-3pt}
\end{equation}
where $P_s$ is the transmit power, $n\sim\mathcal{CN}(0, N_0)$ is the complex
additive white Gaussian noise (AWGN), $PL_d$ and $PL_{k,n}$ denote the path losses of the direct link and the cascaded link via the $n$-th element of the $k$-th R-UAV, respectively.\vspace{-10pt}
\subsection{Outdated and Imperfect CSI Model}
In practical RIS-aided UAV networks, perfect CSI is generally unavailable due to two sequential impairments. First, channel estimation becomes imperfect because of thermal noise, hardware limitations, and excessive overhead during the training stage. Second, owing to feedback delay, processing latency, and UAV mobility, this imperfect estimate may become outdated relative to the actual channel at the transmission time. 
We adopt a quasi-static block-fading model, where the channel remains constant within one transmission block and CSI aging is statistically captured through Clarke's fading model. 
Therefore, imperfect and outdated CSI are modeled as a sequential estimated-then-outdated process. 
According to \cite{Eq1imperfect,Eq1out}, the relation between the outdated and instantaneous channel coefficients can be expressed as\vspace{-3pt}
\begin{equation}
    h_d=\widehat{\rho}_d\underbrace{(\widetilde{\rho}_d\widetilde{h}_d+\sqrt{1-\widetilde{\rho}_d^2}\Delta h_d)}_
    {\widehat{h}_d}+\sqrt{1-\widehat{\rho}_d^2}e_d,\vspace{-7pt}
\end{equation}
and\vspace{-7pt}
\begin{equation}
    \bm{\chi}_k=\widehat{\rho}_k\underbrace{(\widetilde{\rho}_k\widetilde{\bm{\chi}}_k+\sqrt{1-\widetilde{\rho}_k^2}\Delta\bm{\chi}_k)}_{\widehat{\bm{\chi}}_k}+\sqrt{1-\widehat{\rho}_k^2}\mathbf{e}_k,\vspace{-5pt}
\end{equation}
where $\widehat{h}_d$ and ${\widehat{\bm{\chi}}}_k$ denote the estimated outdated channel coefficients of the direct and cascaded links, respectively. The terms $e_d\sim\mathcal{CN}(0,\widehat{\sigma}_d^2)$ and $\mathbf{e}_k=[e_{k,1},...,e_{k,N}]\in \mathbb{C}^{1\times N}$ with $e_{k,n}\sim\mathcal{CN}(0,\widehat{\sigma}_{k,n}^2)$ represent error terms due to outdated CSI. 
The correlation coefficients between the actual and outdated channel estimates are given by $\widehat{\rho}_{i}=\mathcal{J}_0 (2\pi f_d T_{d}) (i\in\left\{d,k\right\})$ based on Clarke's fading model, where $\mathcal{J}_0(\cdot)$ is the zeroth-order Bessel function of first kind, $T_{d}$ is the delay between the outdated CSI and instantaneous CSI, and $f_d$ denotes the Doppler spread caused by UAV mobility\cite{model2s-1}. 
The terms $\widetilde{h}_d$ and $\widetilde{\bm{\chi}}_k = [\widetilde{{\chi}}_{k,1},...,\widetilde{{\chi}}_{k,N}]\in\mathbb{C}^{1\times N}$ are the information of direct and cascaded links obtained by channel estimation technique, respectively, while $\Delta h_d\sim \mathcal{CN}(0, \widetilde{\sigma}_d^2)$ and $\Delta\chi_{k,n}\sim\mathcal{CN}(0, \widetilde{\sigma}_{k,n}^2)$ are the corresponding channel estimation error terms, respectively. $\widetilde{\rho}_i \in (0,1] (i\in\left\{d,k\right\})$ is the estimated correlation coefficient, depending on the specific estimation method and the number of pilot symbols used\cite{model1}. 

Although the optimal phase shift can be designed as {$\Phi_{k,n}^{*}=-\angle\widetilde{h}_{d}+\angle \widetilde{g}_{k,n}+\angle \widetilde{h}_{k,n}$} to maximize the received signal power, it is infeasible to implement continuous phase control in practice due to limited hardware computational capabilities of microcontrollers and circuits integrated with RIS\cite{continuousphase}. Therefore, for the $k$-th R-UAV with $q_k$-bit resolution and $2^{q_k}$ discrete phase shift, the optimal phase shift is selected as the closest value from set $\left\{{2\lambda\pi}/{2^{q_k}}\right\}$ with $\lambda \in \left\{0, 1, ..., 2^{q_k}-1\right\}$, which introduces phase error, denoted as $\theta_{k,n}\sim \mathcal{U}\left(-\frac{\pi}{2^{q_k}},\frac{\pi}{2^{q_k}}\right)$\cite{discretephase}. The quantization error $\theta_{k,n}$ causes residual phase misalignment among the reflected components and reduces the ideal coherent combining gain at the UE. 
Consequently, Eq.~\eqref{signal} can be rewritten as
\begin{equation}\vspace{-5pt}
	\begin{aligned}
		y&=\sqrt{{P_{s}}}\left(\underbrace{ \frac{\widehat{\rho}_{d}\widetilde{\rho}_{d}|\widetilde{h}_{d}|}{\sqrt{PL_d}}}_{D} +\underbrace{ \sum_{k}^{K}\sum_{n}^{N}\frac{\widehat{\rho}_{k,n}\widetilde{\rho}_{k,n}|\widetilde{\chi}_{k,n}| e^{-j\theta_{k,n}}}{\sqrt{PL_{k,n}}}}_{A}
		\right.\\
		&+\underbrace{\frac{\widehat{\rho}_{d}\sqrt{1-\widetilde{\rho}_d^2}\Delta h_{d}}{\sqrt{PL_d}} + \sum_{k}^{K}\sum_{n}^{N}\frac{\widehat{\rho}_{k,n}\sqrt{1-\widetilde{\rho}_{k,n}^2}\Delta \chi_{k,n}e^{j\Phi_{k,n}}}{\sqrt{PL_{k,n}}}}_{G}\\
		&\left.+ \underbrace{\left(\frac{\sqrt{1-\widehat{\rho}_{d}^2}e_d}{\sqrt{PL_d}} + \sum_{k}^{K}\sum_{n}^{N}\frac{\sqrt{1-\widehat{\rho}_{k,n}^2}e_{k,n}}{\sqrt{PL_{k,n}}} \right)}_{O} \right)x + n,
	\end{aligned}\label{siganl2}
\end{equation} 
where the first and second terms are the channel coefficients of direct and cascaded links with optimal phase shift, the third and fourth terms denote the channel coefficients with estimated and outdated error, respectively. For simplicity, we define $X = D + A$ as the composite coefficient of the desired channel,  and $E = G + O$ as the aggregate erroneous channel component.

\begin{remark}\label{complexgaussian}
$G$ and $O$ are linear combinations of independent complex Gaussian variables and follow circularly symmetric complex Gaussian (CSCG) distributions, i.e., $G\sim\mathcal{CN}(0, \sigma_{G}^2)$ and  $O\sim\mathcal{CN}(0, \sigma_{O}^2)$, with variances $\sigma_{G}^2=\frac{\widehat{\rho}_{d}^2(1-\widetilde{\rho}_{d}^2)\widetilde{\sigma}_{d}^2}{PL_d}+\sum_{k}^{K}\sum_{n}^{N}\frac{\widehat{\rho}_{k,n}^2(1-\widetilde{\rho}_{k,n}^2)\widetilde{\sigma}_{k,n}^2}{PL_{k,n}}$ and $\sigma_{O}^2=\frac{(1-\widehat{\rho}_{d}^2)\widehat{\sigma}_{d}^2}{PL_d}+\sum_{k}^{K}\sum_{n}^{N}\frac{(1-\widehat{\rho}_{k,n}^2)\widehat{\sigma}_{k,n}^2}{PL_{k,n}}$, respectively. Correspondingly, $E\sim\mathcal{CN}(0,\sigma_{E}^2)$, where $\sigma_{E}^2=\sigma_{G}^2+\sigma_{O}^2$.\footnote{The independent CSCG modeling of CSI errors is adopted as a tractable second-order surrogate for performance analysis. It does not exclude the possibility of correlated or non-Gaussian errors in high-mobility UAV scenarios.}
\end{remark}

\vspace{-9pt}
\subsection{Composite Fading Channel Model and Path Loss Model}
For the R-UAV--UE and UAV-BS--UE links, the path loss model based on the statistical parameters of the underlying terrestrial environment is adopted \cite{LAP}. Based on \cite[Eqs.~(2) and (7)]{tcomF}, the path loss (measured in dB) of the cascaded link over the $k$-th R-UAV, denoted by $PL_k^{\text{dB}}$, can be expressed as\vspace{-2pt}
\begin{equation}\label{plkmodel}
\begin{aligned}
{PL_k^\text{dB}}=10\alpha\log_{10}\left(\frac{d_{1,k}}{\lambda}\right)&+20\log_{10}\left(\frac{4\pi d_{2,k}}{\lambda}\right)\\&+\frac{(\eta_{\text{LoS}}-\eta_{\text{NLoS}})}{1+ae^{-b(\theta_k-a)}}+\eta_{\text{NLoS}},
\end{aligned}\vspace{-2pt}
\end{equation}
where $\alpha$ is the path loss exponent, $d_{1,k}=\left\lVert \boldsymbol{\rm q_b}-\boldsymbol{\rm q_k} \right\rVert$ is the distance between the UAV-BS and the $k$-th R-UAV. $d_{2,k}=\left\lVert \boldsymbol{\rm q_k}-\boldsymbol{\rm q_0} \right\rVert$ and $\theta_k$ are the distance and elevation angle between the $k$-th R-UAV and UE, respectively. The S-curve parameters $a$ and $b$ are specific to the simulated environment \cite[Tables I-II]{LAP}. The terms $\eta_{\text{LoS}}$ and $\eta_{\text{NLoS}}$ represent excessive path loss of LoS and NLoS propagation groups, respectively. 
The path loss of the direct link can be expressed as
\begin{equation}\label{pldmodel}
PL_d^{\text{dB}}={20\log_{10}\left(\frac{4\pi d }{\lambda}\right)}+\frac{(\eta_{\text{LoS}}-\eta_{\text{NLoS}})}{1+ae^{-b(\theta_d-a)}}+\eta_{\text{NLoS}},\vspace{-2pt}
\end{equation}
where $d=\left\lVert \boldsymbol{\rm q_b}-\boldsymbol{\rm q_0} \right\rVert$ is the distance between the UAV-BS and the UE, and $\theta_d$ is the corresponding elevation angle. The corresponding path loss terms in the linear scale $PL_k$ and $PL_d$ can be directly converted from $PL_k^{\text{dB}}$ and $PL_d^{\text{dB}}$, respectively\footnote{{Due to the small physical aperture of each RIS relative to the propagation distances, the path-loss variation within the same R-UAV is neglected, yielding $PL_{k,n}\simeq PL_k$, $\forall n$.}}.

The UAV-BS--R-UAV channel is modeled as a normalized LoS-dominant aerial link for tractability, where the large-scale fading is captured as $10\alpha\log_{10}(d_{1,k}/\lambda)$ in $PL_k$ and the small-scale component is idealized as unit modulus, i.e., $|g_{k,n}|=1$. This assumption is adopted as a benchmark abstraction for high-altitude aerial backhaul links. 
For the UE trapped beneath ruins or in collapsed underground spaces, the channel is dominated by complex terrain and local scatterers, inducing severe composite fading. To accurately characterize such an environment, the modified Fisher-Snedecor $\mathcal{F}$ fading model in~\cite{tcomF} is adopted, which accurately characterizes multi-path fading and shadowing via parameters $m$ and $m_s$, respectively
. Therefore, we consider that the channel envelopes pertaining to the UE, i.e., $|\widetilde{h}_d|$ and $|\widetilde{\chi}_{k,n}|$, undergo the modified Fisher-Snedecor $\mathcal{F}$ fading, the PDFs of which are given by
\vspace{-5pt}\begin{equation}\label{hdPDF}
	f(|\widetilde{h}_d|)=\frac{2(m^d\Omega_{s}^d)^{m^d}(m_{s}^d\Omega_{m}^d)^{m_{s}^d}{|\widetilde{h}_d|}^{2m^d-1}}{B(m^d,m_{s}^d)[m^d\Omega_{s}^d{|\widetilde{h}_d|}^2+m_{s}^d\Omega_{m}^d]^{m^d+m_{s}^d}},\vspace{-5pt}
\end{equation}
and\vspace{-5pt}
\begin{equation}\label{hrisPDF}
	\begin{aligned}
	&f(|\widetilde{{\chi}}_{k,n}|)=\\
	&\frac{2(m^{k,n}\Omega_{s}^{k,n})^{m^{k,n}}(m_{s}^{k,n}\Omega_{m}^{k,n})^{m_{s}^{k,n}}{|\widetilde{\chi}_{k,n}|}^{2m^{k,n}-1}}{B(m^{k,n},m_{s}^{k,n})[m^{k,n}\Omega_{s}^{k,n}{|\widetilde{\chi}_{k,n}|}^2+m_{s}^{k,n}\Omega_{m}^{k,n}]^{m^{k,n}+m_{s}^{k,n}}},
	\end{aligned}\vspace{-4pt}
\end{equation}
respectively, where $B(a, b)$ is the beta function with $B(a,b) =\frac{\Gamma(a)\Gamma(b)}{\Gamma(a+b)}$\cite[Eq.~(8.384.1)]{2007Table} and $\Gamma(b) = \int_{0}^{\infty}{t^{b-1}e^{-t}}dt$ is the gamma function \cite[ Eq.~(8.310)]{2007Table}. Here, $m^{i}$ and $m_s^{i}~ (i\in\left\{d, (k,n)\right\})$ denote the multi-path fading and shadowing parameters. Specifically, $m^{i}\rightarrow 0^+$ and $m_s^{i} \rightarrow 1^+$ indicate harsher channel conditions with severer multipath/shadowing, whereas $m^{i}\rightarrow \infty$ and $m_s^i\rightarrow \infty$ correspond to better propagation conditions. Notably, as $m_s^{i} \rightarrow \infty$, the shadowing effect vanishes and the distribution reduces to Nakagami-$m$ fading\cite{Fchannel}. This makes the $\mathcal{F}$ model especially useful for capturing a wide range of emergency-environment conditions. The terms $\Omega_m^{i}$ and $\Omega_{s}^{i}$ represent the mean power of the multi-path and shadowing components, respectively, resulting in the mean power of the composite signal envelope as $\Omega^{i}=\Omega_m^{i}\Omega_s^{i}$.

\vspace{-13pt}
\subsection{The Unified SNR Framework ($\zeta$-Model)}
In the presence of outdated and imperfect CSI, existing works typically categorize the erroneous components into either effective channel gain or adverse interference, making the corresponding SNR models as isolated receiver assumptions rather than interpreted within a common analytical framework.
To address this issue, we propose a unified 
SNR framework, referred to as the $\zeta$-Model, which organizes these CSI-error treatments into a common analytical family under the same residual-error budget.

Specifically, $E$ can be interpreted as the aggregate residual uncertainty after the initial CSI acquisition stage, capturing the remaining impairment caused by outdated CSI and estimation imperfections. To characterize the receiver's capability to partially exploit this residual term through additional receiver-side information $\mathcal{S}$\footnote{
The side information $\mathcal{S}$ abstracts auxiliary receiver-side information available beyond the initial CSI acquisition, such as position information, map/digital-twin features, or sensing-assisted context information\cite{S2,S1,S4}.}
, we decompose
$E$ as $E=E_\zeta+\varepsilon_\zeta$, 
where $E_\zeta \triangleq \mathbb{E}[E \mid \mathcal{S}]$ denotes the receiver-exploitable component, and $\varepsilon_\zeta  \triangleq E - \mathbb{E}[E \mid \mathcal{S}]$ is the remaining unexplained residual term. By the orthogonality principle of MMSE estimation, the unexplained residual $\varepsilon_\zeta$ is orthogonal to any $\mathcal S$-measurable function, leading to $\mathbb{E}[E_\zeta \varepsilon_\zeta^\ast]=0$, which further yields $\mathbb{E}[|E|^2]=\mathbb{E}[|E_\zeta|^2]+\mathbb{E}[|\varepsilon_\zeta|^2]$.
Then, we define the residual-error exploitation coefficient $\zeta\in[0,1]$ 
as follows:
\begin{equation}\label{zeta_def}
\zeta \triangleq \frac{\mathbb{E}\left[|E_\zeta|^2\right]}{\mathbb{E}[|E|^2]}, ~~~~1-\zeta\triangleq\frac{\mathbb{E}\left[|\varepsilon_\zeta|^2\right]}{\mathbb{E}[|E|^2]}.\vspace{-5pt}
\end{equation}
For a concrete refinement/prediction method, an effective operating point of $\zeta$ can be calibrated through the relative NMSE improvement \cite{S4}
\(
\zeta_{\mathrm{eff}}=1-\mathrm{NMSE}_{\mathrm{refined}}/\mathrm{NMSE}_{\mathrm{initial}}
\).
To obtain a tractable $\zeta$ analytical family, we adopt the following receiver-side Gaussian-equivalent model as\vspace{0pt}
\begin{equation}\label{Gaussian-equivalent model}
	\begin{bmatrix}E_\zeta\\
		\varepsilon_\zeta
	\end{bmatrix}
	\sim
	\mathcal{CN}\!\left(
	\mathbf 0,
	\begin{bmatrix}
		\zeta \sigma_E^2 & 0\\
		0 & (1-\zeta)\sigma_E^2
	\end{bmatrix}
	\right).\vspace{-2pt}
\end{equation}

Based on the above decomposition, the generalized SNR, denoted by $\gamma_\zeta$, can be represented as follows:\vspace{-3pt}
\begin{equation}\label{zetamodel_defination}
	\begin{aligned}
	\gamma_\zeta &=  \frac{P_s |X+E_\zeta|^2}{P_s |\varepsilon_\zeta|^2 + N_0}= \frac{\bar{\gamma} |U|^2}{\bar{\gamma} |\varepsilon_\zeta|^2 + 1},
	\end{aligned}\vspace{-2pt}
\end{equation}
where $\bar{\gamma}={P_s}/{N_0}$ is the average SNR, and the composite useful term $U= X+ E_\zeta$ is defined for simplicity.

In the following, we commence by introducing the Information-Theoretic Model (ITM) and Engineering Baseline Model (EBM) as two special cases at the boundaries of $\zeta$, followed by the  Simplified Engineering Model (SEM) as a tractable approximation of the EBM. 

\subsubsection{Information-Theoretic Model (ITM)} 
{The idealized boundary case $\zeta = 1$ corresponds to a receiver model in which the entire residual term is assumed exploitable, 
yielding $\mathbb{E}_\zeta=E$ and $\varepsilon_\zeta=0$\cite{imperfectmodel1,imperfectoutmodel1,model1}.} Accordingly, the SNR at the UE can be expressed as follows:
\begin{equation}
	\gamma_1= \frac{P_s |X + E|^2}{N_0} = \bar{\gamma} |Z|^2,\vspace{-5pt}
\end{equation}
where $Z = X+E$ is the overall channel statistic. From an information-theoretic benchmark perspective, if the receiver can perfectly exploit the instantaneous effective channel realization, the achievable rate is determined by the overall channel statistic $Z=X+E$. Therefore, the ITM serves as an ideal full-exploitation upper benchmark within the considered residual-error budget.

{\subsubsection{Engineering Baseline Model (EBM)} When $\zeta=0$, the side information $\mathcal{S}$ is statistically uninformative about $E$, leading to ${E}_\zeta=0$ and $\varepsilon_\zeta=E$.} Therefore, the receiver is assumed to have no knowledge of $E$ and treats it entirely as detrimental and uncorrelated interference that degrades the decoding performance\cite{Eq1imperfect,model2s-1,model2s-2}. The instantaneous SNR at the UE can be expressed as follows:
\begin{equation}\begin{aligned}\gamma_0 = \frac{P_s |X|^2}{P_s |E|^2 + N_0} = \frac{\bar{\gamma} |X|^2}{\bar{\gamma} |E|^2 + 1}.\end{aligned}\vspace{-3pt}
\end{equation}
This model serves as a conservative engineering baseline where the residual uncertainty is fully treated as harmful interference. Hence, it is useful for characterizing the robustness floor of the system under outdated and imperfect CSI. 

\subsubsection{Simplified Engineering Model (SEM)} 
Analyzing the performance of EBM is mathematically complex due to the presence of the instantaneous random variable $|E|^2$ in the denominator. To facilitate a tractable analysis, a common approach in \cite{model2s-1,model2s-2} is proposed to approximate the instantaneous interference power with its expected value, which can be expressed as follows: 
\begin{equation}\label{Model2gamma_simple}
	\gamma_0^s=\frac{P_s|X|^2}{\sigma_{\text{ef}}^2}={\bar{\gamma}_{_\text{ef}}|X|^2},\vspace{-3pt}
\end{equation}
where $\sigma_{\text{ef}}^2=\mathbb{E}[P_s|E|^2+N_0]=P_s\sigma_{E}^2+N_0$ represents the effective noise power combining the interference and AWGN, and  $\bar{\gamma}_{_\text{ef}}={P_s}/{\sigma_{\text{ef}}^2}$  denotes the effective transmit SNR. 
\begin{remark}
The novelty of the proposed framework lies in recasting the aggregated residual impairment into a receiver-oriented exploitable/unexplained decomposition with respect to additional side information. This enables the cases of $\zeta=0$ and $\zeta=1$ to be interpreted as two boundary treatments of the same residual-error budget, rather than as unrelated modelling choices adopted in different prior works.
\end{remark}

\vspace{-5pt}
\section{Statistical Analysis of SNR}\label{section3}\vspace{-5pt}
In this section, we first characterize the statistical properties of the received SNR under the unified $\zeta$-Model based on gamma moment matching, and then derive the distributions for the boundary cases, i.e., ITM, EBM, and SEM. 
Following the derivations, we employ KL divergence to verify the global goodness-of-fit, and then use the MRE to specifically assess the approximation accuracy of the tail distributions.
\vspace{-5pt}
\subsection{Statistical Analysis of SNR}
\subsubsection{$\zeta$-Model under Different Shadowing Regimes}
Based on \textit{Remark} \ref{complexgaussian}, $|E|^2$ obeys an exponential distribution with parameter $\sigma_{E}^2$. 
Hence, according to the receiver-side Gaussian-equivalent model in 
Eq.~\eqref{Gaussian-equivalent model}, $|\varepsilon_\zeta|^2$ follows an exponential distribution with parameter $\sigma_{\varepsilon_\zeta}^2=(1-\zeta)\sigma_{E}^2$, the PDF of which is given by \vspace{-5pt}
\begin{equation}\label{PDF_varepsilon}
	f_{_{|\varepsilon_\zeta|^2}}(\varepsilon)=\frac{1}{\sigma_{\varepsilon_\zeta}^2}\exp\left(-\frac{\varepsilon}{\sigma_{\varepsilon_\zeta}^2}\right), \quad \varepsilon \geq 0.
\end{equation}
The corresponding CDF of $|\varepsilon_\zeta|^2$, denoted by $F_{_{|\varepsilon_\zeta|^2}}(\varepsilon)$, can be expressed as follows:\vspace{-5pt}
\begin{equation}
	F_{_{|\varepsilon_\zeta|^2}}(\varepsilon) =1-\exp\left(-\frac{\varepsilon}{\sigma_{\varepsilon_\zeta}^2}\right),\quad \varepsilon\geq 0.\vspace{-5pt}
	\label{CDF_V}
\end{equation}

Due to the composite structure of the numerator term $|U|^2$, deriving its exact distribution is analytically intractable. The moment matching method works well for positive RVs with unimodal and fast-decaying PDFs\cite{R3-GAMMA}, while the Gamma distribution, as a Type-III Pearson distribution, is widely used to match such RVs through their first two moments\cite{gamma_momentmatch}. Therefore, we employ Gamma moment matching to obtain tractable approximate PDFs/CDFs. 

However, the admissible moment order of the modified Fisher-Snedecor-\(\mathcal F\) envelope depends explicitly on \(m_s\). Specifically, for the power-domain matching of \(|U|^2\), the first two moments of \(|U|^2\) are required, which in turn involve the fourth-order moments of the underlying channel terms. For the amplitude-domain matching of \(|U|\), the mean and variance of \(|U|\) are required.  According to Eq.~\eqref{Eh}, the \(i\)-th moment of a Fisher-Snedecor-\(\mathcal F\) RV exists only when \(m_s>i/2\). Hence, the fourth-order moment diverges for \(m_s\leq 2\), while the second-order moment diverges for \(m_s\leq 1\). 
Therefore, we adopt a regime-dependent construction: power-domain matching of \(|U|^2\) for \(m_s>2\), and amplitude-domain matching of \(|U|\) for \(1<m_s\le 2\) in the following. 



$\mathit{a)}~ $For $m_s>2$, we use the gamma moment-matching approach to approximate $|U|^2$. Accordingly, the PDF of $|U|^2$, denoted by $f_{{|U|^{2}}}(x)$, can be expressed as follows:\vspace{-5pt}
\begin{equation}
	f_{{|U|^{2}}}(u) \approx \frac{u^{k_{u_2}-1}}{\Gamma(k_{u_2})\theta_{u_2}^{k_{u_2}}}  \exp\left({-\frac{u}{\theta_{u_2}}}\right),  \vspace{-3pt}
	\label{PDF_U2}
\end{equation}
where $k_{u_2}={\mu_{|U|^2}^2}/{\sigma_{|U|^2}^2}$ and $\theta_{u_2}={\sigma_{|U|^2}^2}/{\mu_{|U|^2}}$ are the shape and scale parameters, respectively. 
The corresponding CDF of $|U|^2$, denoted by $F_{{|U|^{2}}}(u)$, can be expressed as 
\begin{equation}
	F_{{|U|^{2}}}(u) \approx \frac{1}{\Gamma(k_{u_2})} \gamma(k_{u_2}, \frac{u}{\theta_{u_2}}), 
	\label{CDF_U2}
\end{equation}
where $\gamma(c, x)=\int_{0}^{x}t^{c-1}e^{-t}dt$ denotes the lower incomplete gamma function.

$\mathit{b)}~ $For $1<m_s\leq 2$, we approximate $|U|$ with gamma moment-matching. 
Accordingly, the PDF of $|U|$, denoted by $f_{|U|}(u)$, can be expressed as follows:
\begin{equation}
	f_{|U|}(u) \approx \frac{u^{k_{u_1}-1}}{\Gamma(k_{u_1})\theta_{u_1}^{k_{u_1}}}  \exp\left({-\frac{u}{\theta_{u_1}}}\right), \quad u \geq 0,\vspace{-3pt}
	\label{PDF_U}
\end{equation}
where $k_{u_1}={\mu_{|U|}^2}/{\sigma_{|U|}^2}$ and $\theta_{u_1}={\sigma_{|U|}^2}/{\mu_{|U|}}$. 
The corresponding CDF of $|U|$, denoted by $F_{|U|}(u)$, can be expressed as 
\begin{equation}
	F_{|U|}(u) \approx \frac{1}{\Gamma(k_{u_1})} \gamma(k_{u_1}, \frac{u}{\theta_{u_1}}), \quad u \geq 0.\vspace{-5pt}
	\label{CDF_U}
\end{equation}
The statistical moments $\mu_{|U|}$, $\sigma_{|U|}^2$, $\mu_{|U|^2}$ and $\sigma_{|U|^2}^2$ are all presented in Appendix~\ref{AppendixA}.

Under the receiver-equivalent decoupling surrogate, \(U\) is treated as independent of the denominator \(\varepsilon_\zeta\)\footnote{This step is a stronger receiver-equivalent decoupling surrogate, motivated by the analytical decoupling commonly adopted in classical imperfect-CSI analyses \cite{Eq1imperfect,X-Eindependence}, where the desired channel term and the Gaussian error component are treated as independent to obtain tractable SNR characterization.}. Therefore, the CDF of $\gamma_\zeta$, denoted by $F_{\gamma_\zeta}(\gamma)$, can be derived in Eq.~\eqref{CDF_GAMMAZETA} as shown on the bottom of this page, which involves closed-form and finite-sum approximate expressions under $m_s>2$ and $1<m_s\leq 2$, respectively. $N_q$ is the truncation order of the Gauss-Laguerre quadrature, $y_i$ is the $i$-th root of the Laguerre polynomial, and $w_i$ is the corresponding weight factor. 
\begin{IEEEproof}
	See Appendix~\ref{AppendixB}.
\end{IEEEproof}

{\begin{remark}
The intermediate-\(\zeta\) characterization is based on a receiver-equivalent analytical surrogate. Specifically, the conditional-mean decomposition in Eq.~\eqref{zeta_def} provides theoretical interpretation of \(\zeta\), while the jointly Gaussian independent representation in Eq.~\eqref{Gaussian-equivalent model} specifies a tractable statistical model for \((E_\zeta,\varepsilon_\zeta)\). The moment-level closure \(X\perp E_\zeta\) adopted in Appendix A is for the moment matching of $U$, while \(U\perp \varepsilon_\zeta\) is for tractable SNR characterization. These assumptions enable performance characterization under a unified residual-error budget. If a specific receiver/refinement architecture induces statistical coupling among these terms, additional mixed moments would be required.
\end{remark}}

\begin{figure*}[b]
				\rule[0pt]{18.05cm}{0.05em}  
\begin{equation}\label{CDF_GAMMAZETA}
	F_{\gamma_\zeta}(\gamma)=
		\begin{cases}
		\begin{aligned}
		&1-\frac{1}{\Gamma(k_{u_2})}\left[\Gamma(k_{u_2},\frac{\gamma }{\bar{\gamma}\theta_{u_2}}) - \exp\left(\frac{1}{\bar{\gamma}\sigma_{\varepsilon_\zeta}^2}\right)\left(\frac{\sigma_{\varepsilon_\zeta}^2\gamma}{\theta_{u_2}+\sigma_{\varepsilon_\zeta}^2\gamma}\right)^{k_{u_2}}\Gamma(k_{u_2},\frac{\theta_{u_2}+\sigma_{\varepsilon_\zeta}^2\gamma }{\bar{\gamma}\theta_{u_2}\sigma_{\varepsilon_\zeta}^2})\right],  ~~~~m_s>2;\\
		&\frac{1}{\Gamma(k_{u_1})} \sum_{i=1}^{N_{q}} w_{i} \gamma\left(k_{u_1}, \frac{\sqrt{\gamma ( \sigma_{\varepsilon_\zeta}^2 y_i +\frac{1}{\bar{\gamma}})}}{\theta_{u_1}}\right), ~~~~~~~~~~~~~~~~~~~~~~~~~~~~~~~~~~~~~~~~~~~~~~~~~~{1<m_s \leq 2}.
		\end{aligned}
		\end{cases}
\end{equation}
\end{figure*}

\begin{figure}\vspace{-5pt}
	\centering
	{\hspace{-5pt}
	\includegraphics[width=0.9\columnwidth]{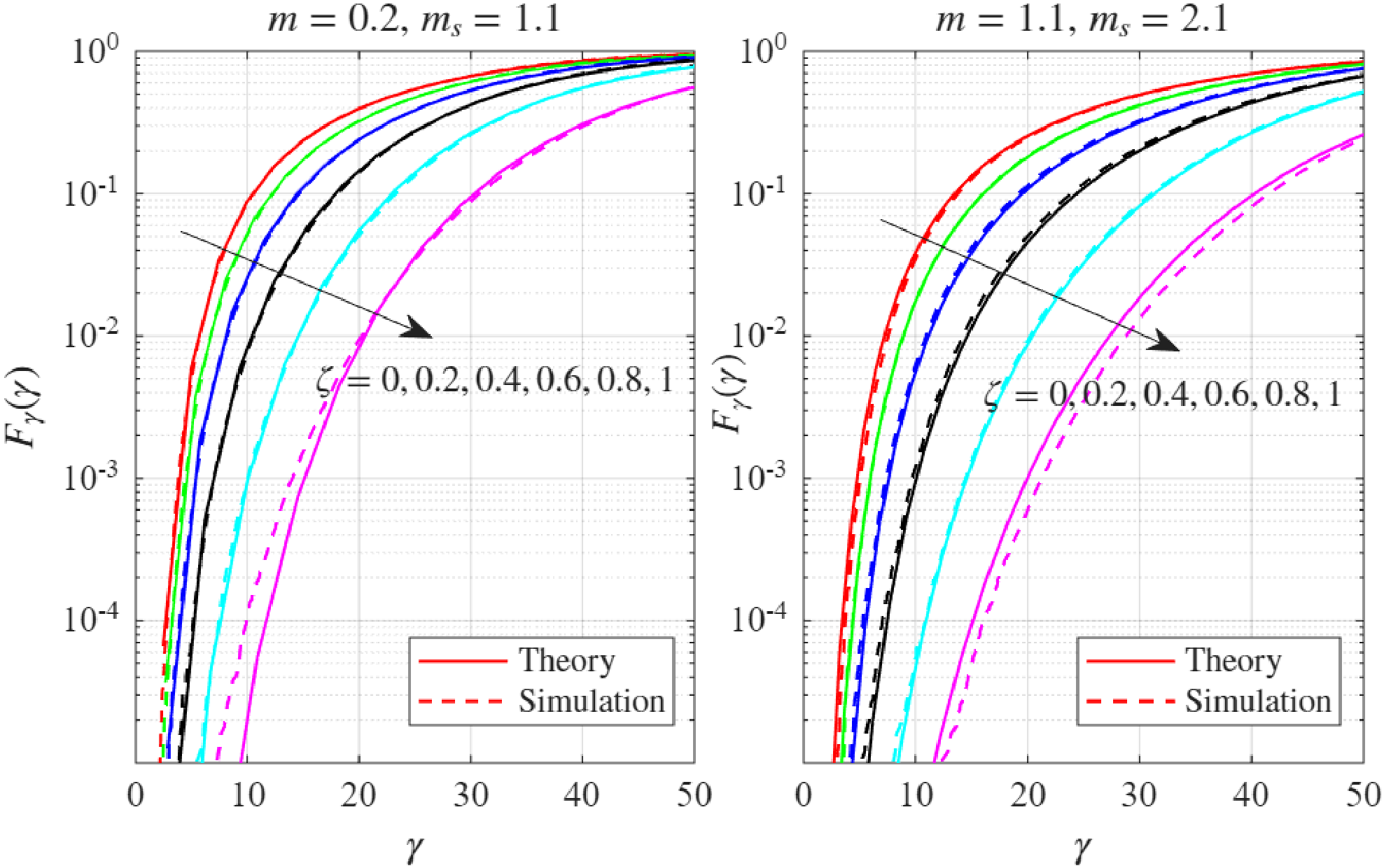}}
	\vspace{-8pt}\caption{Simulated and theoretical CDF for $\zeta$-Model under {$1<m_s \leq 2$} and $m_s>2$, respectively.}
	\label{CDF}\vspace{-15pt}
\end{figure}

Since directly obtaining the PDF expression of $\gamma_\zeta$ by differentiating Eq.~\eqref{CDF_GAMMAZETA} is mathematically intractable, 
we further focus on establishing the distributions in closed-form for the two boundary cases with $\zeta=0$ and $\zeta=1$. This strategy not only yields analytically tractable benchmark results, but also provides an interpretable performance span between the pessimistic and optimistic receiver treatments. The statistical gap between $\gamma_1$ and $\gamma_0$ can therefore serve as a critical metric to quantify the potential gain with different treatments of the residual error components. 
\subsubsection{The Boundary Cases} \label{boundary_distribution}
We also distinguish two analytical regimes based on different $m_s$ in the following. 
For the ITM with $m_s > 2$, $|Z|^2 \sim \mathrm{Gamma}(k_{z_2},\theta_{z_2}$) can be obtained following the same derivation in Section III-A, where $k_{z_2}={\mu_{|Z|^2}^2}/{\sigma_{|Z|^2}^2}$ and $\theta_{z_2}={\sigma_{|Z|^2}^2}/{\mu_{|Z|^2}}$ can be calculated based on Appendix~\ref{AppendixA} by setting $\zeta=1$. Likewise, for {$1< m_s \leq 2$}, $|Z|\sim \mathrm{Gamma}(k_{z_1},\theta_{z_1})$, where $k_{z_1}={\mu_{|Z|}^2}/{\sigma_{|Z|}^2}$ and $\theta_{z_1}={\sigma_{|Z|}^2}/{\mu_{|Z|}}$ are also given in Appendix~\ref{AppendixA}. Therefore, after some variable changes, the PDF and CDF of $\gamma_1$, denoted by $f_{\gamma_1}(\gamma)$ and $F_{\gamma_1}(\gamma)$, can be expressed respectively as follows:\vspace{-5pt}
\begin{equation}\label{pdfgamma1}
f_{\gamma_1}(\gamma) =
	\begin{cases}
		\begin{aligned}
		&
		\frac{{\gamma}^{k_{z_2}-1}}{\Gamma(k_{z_2})(\bar{\gamma}\theta_{z_2})^{k_{z_2}}}  \exp\left({-\frac{\gamma}{\bar{\gamma}\theta_{z_2}}}\right),\quad m_s > 2;\\
		&
		\!\frac{{\gamma}^{\frac{k_{z_1}}{2}-1}}{2\Gamma(k_{z_1})\theta_{z_1}^{k_{z_1}}\bar{\gamma}^{\frac{k_{z_1}}{2}}}\!  \exp\!\left(\!{-\frac{1}{\theta_{z_1}}\sqrt{\frac{\gamma}{\bar{\gamma}}}}\right),\! {1 <m_s \leq 2},\vspace{-5pt}
		\end{aligned}
		\end{cases}
\end{equation}
and\vspace{-5pt}
\begin{equation}\label{cdfgamma1}
F_{\gamma_1}(\gamma) =
\begin{cases}
		\begin{aligned}
&\frac{1}{\Gamma(k_{z_2})} \gamma\left(k_{z_2}, \frac{\gamma}{{\theta_{z_2}}\bar{\gamma}}\right), \quad m_s >2;\\
&\frac{1}{\Gamma(k_{z_1})} \gamma(k_{z_1}, \frac{1}{\theta_{z_1}}\sqrt{\frac{\gamma}{\bar{\gamma}}}), \quad {1 < m_s \leq 2}.\vspace{-5pt}
		\end{aligned}
\end{cases}
\end{equation}

For the EBM and SEM with $m_s > 2$, $|X|^2 \sim \mathrm{Gamma}(k_{x_2}, \theta_{x_2})$ can be derived, where $k_{x_2}={\mu_{|X|^2}^2}/{\sigma_{|X|^2}^2}$ and $\theta_{x_2}={\sigma_{|X|^2}^2}/{\mu_{|X|^2}}$ can be calculated from Eqs.~\eqref{EX2} and \eqref{EX4} in Appendix \ref{AppendixA}. Similarly, for the case of $1<m_s \leq 2$, $|X| \sim \mathrm{Gamma}(k_{x_1}, \theta_{x_1})$ can be derived, where $k_{x_1}={\mu_{|X|}^2}/{\sigma_{|X|}^2}$ and $\theta_{x_1}={\sigma_{|X|}^2}/{\mu_{|X|}}$ can also be obtained from Eq.~\eqref{EU}. Correspondingly, for the SEM, the PDF and CDF of ${\gamma^0_s}$, denoted by $f_{\gamma_0^s}(\gamma)$ and $F_{\gamma_0^s}(\gamma)$, can be directly obtained as \vspace{-5pt}
\begin{equation}\label{pdfgamma2s}
		f_{\gamma_0^s}(\gamma) 
		\!=\!
		\begin{cases}\!
		\begin{aligned}
		&\frac{{\gamma}^{k_{x_2}-1}}{\Gamma(k_{x_2})(\bar{\gamma}_{_\text{ef}}\theta_{x_2})^{k_{x_2}}}  \exp\left({-\frac{\gamma}{\bar{\gamma}_{_\text{ef}}\theta_{x_2}}}\right), m_s > 2;\\
		&\!\frac{\gamma^{\frac{k_{x_1}}{2}-1}}{2\Gamma(k_{x_1})\theta_{x_1}^{k_{x_1}}{\bar{\gamma}_{_\text{ef}}}^{\frac{k_{x_1}}{2}}}\!\exp\!\left(\!{\!-\frac{1}{\theta_{x_1}}\!\sqrt{\frac{\gamma}{\bar{\gamma}_{_\text{ef}}}}}\right)\!,  {1<m_s\le 2},\vspace{-5pt}
	\end{aligned}
	\end{cases}
\end{equation}
and \vspace{-5pt}
\begin{equation}\label{cdfgamma2s}
		F_{\gamma_0^s}(\gamma) 
		=
		\begin{cases}
		\begin{aligned}
		&\frac{1}{\Gamma(k_{x_2})} \gamma\left(k_{x_2}, \frac{\gamma}{{\theta_{x_2}}\bar{\gamma}_{_\text{ef}}}\right), \quad m_s >2;	\\
		&\frac{1}{\Gamma(k_{x_1})} \gamma\left(k_{x_1}, \frac{1}{\theta_{x_1}}\sqrt{\frac{\gamma}{\bar{\gamma}_{_\text{ef}}}}\right), \quad {1<m_s\le 2},\vspace{-5pt}
	\end{aligned}
	\end{cases}
\end{equation}
respectively. Nevertheless, for the EBM, the CDF expression of  $\gamma_0$ can only be obtained by replacing $\left\{k_{u_2},k_{u_1}\right\}$ with $\left\{k_{x_2},k_{x_1}\right\}$, $\left\{\theta_{u_2},\theta_{u_1}\right\}$ with $\left\{\theta_{x_2},\theta_{x_1}\right\}$, and $\sigma_{\varepsilon_\zeta}^2$ with $\sigma_{E}^2$ in Eq.~\eqref{CDF_GAMMAZETA}, respectively.

Figure~\ref{CDF} illustrates the CDF of $\gamma_\zeta$ for different $\zeta$ under the two shadowing regimes $m_s>2$ and $1<m_s\le 2$. In the considered parameter settings, as $\zeta$ increases from $0$ to $1$, CDF curves generally shift toward the higher SNR region, indicating the performance gain achievable as the receiver transitions from treating error as interference to fully exploiting it. 
Moreover, the discrepancy between theoretical and simulated curves becomes more visible in the left-tail region, particularly for larger $\zeta$ and milder channel conditions. 
The reason is that better residual error exploitation and milder fading/shadowing conditions shift the CDF to the right, so that the same fixed interval of $\gamma$ effectively probes a much more extreme left-tail portion of the SNR distribution. 

\vspace{-5pt}\subsection{Statistical Validation of Approximation Accuracy}\label{verification}
To evaluate the accuracy of the approximations for the proposed model, we first employ a statistical metric for global goodness-of-fit\cite{KL} to quantify the distance between the reference empirical distribution, $f_\text{emp}(x)$, and the theoretical approximation, $f_\text{app}(x)$.    
The KL divergence $D_{\text{KL}}$ is given by \vspace{-2pt} 
\begin{equation}\label{KL}
	D_{\text{KL}}=\int_{0}^\infty f_\text{emp}(x)\log\left(\frac{f_\text{emp}(x)}{f_\text{app}(x)}\right)dx,\vspace{-3pt} 
\end{equation}
where smaller values indicate a higher fitting accuracy. 

However, according to the Large Deviations Theory, the CLT-based approximation may not be accurate if $x$ is far from $\mathbb{E}[x]$. It implies $D_{\text{KL}}$ is often dominated by the distribution's mean and is less sensitive to deviations in the tail. 
To further assess the approximation accuracy in the tail of the distribution
, we adopt a quantile-aligned tail error metric {MRE}. Specifically, we consider a fixed probability grid $
\mathcal{R}=\{r_\ell\mid \tau_\text{lb}\le r_\ell\le \tau_\text{ub}\}$,
where $\tau_\text{ub}$ is determined by the target reliability of interest, $\tau_\text{lb}$ acts as a statistical confidence floor to ensure numerical
stability of empirical CDF, and the probability window $[\tau_\text{lb},\tau_\text{ub}]$ directly corresponds to the reliability range of engineering interest. 
For each $r_\ell\in\mathcal{R}$, the corresponding empirical quantile is  $\gamma_\ell = F_{\text{emp}}^{-1}(r_\ell)$.
Then, the MRE, denoted by $\bar{\epsilon}_F$, is defined to evaluate the relative CDF mismatch over a fixed probability window, expressed as 
\begin{equation}
\bar{\epsilon}_F
=
\frac{1}{|\mathcal{R}|}
\sum_{r_\ell\in\mathcal{R}}
\frac{\left|F_{\text{app}}(\gamma_\ell)-r_\ell\right|}{r_\ell}.
\end{equation}
Figure \ref{KL_N} plots the $D_{\text{KL}}$ for the proposed $\zeta$-model versus $N$ and $\zeta$ under two shadowing regimes of $m_s$, where the ITM and SEM boundaries are also included for comparison. In the numerical evaluation, $f_\text{app}(x)$ is obtained from the approximate CDF by numerical differentiation, while $f_\text{emp}(x)$ is estimated from Monte Carlo empirical samples using the same SNR grid. 
It is observed that, at each fixed $N$, $D_{\text{KL}}$ remains nearly flat as $\zeta$ increases and decreases with $N$ for all $\zeta$ for both subfigures, which suggests that the adopted moment-matching framework becomes more accurate for larger RIS sizes and is relatively robust to $\zeta$. 
At the boundary $\zeta=0$, the fitting accuracy of SEM first achieves a comparable or even superior fit to that of $\zeta$-model at
low $N$, but then deteriorates at
large $N$, since the expectation modeling for error term fails as $N$ increases, leading to an
increasing mismatch. 
At the boundary $\zeta=1$, the ITM curve is close to the $\zeta$-model in the left subfigure, while a visible deviation appears in the right subfigure
\footnote{This discrepancy for $m_s>2$ is mainly caused by the numerical ill-conditioning of the unified closed-form as $\sigma_{\varepsilon_\zeta}^2 \to 0$ in Eq.~\eqref{CDF_GAMMAZETA}. It
becomes more visible in small $N$ regime since the exponential and upper incomplete Gamma terms with $\sigma_{\varepsilon_\zeta}^2$ cause a removable-singularity-type limit and loss of numerical accuracy near $\zeta=1$,
whreas the expression for $1<m_s\leq 2$ has a quadrature-based limit and does not exhibit the same trend.}.

\begin{figure}\vspace{-5pt}
	\centering
	{\hspace{5pt}
	\includegraphics[width=0.96\columnwidth]{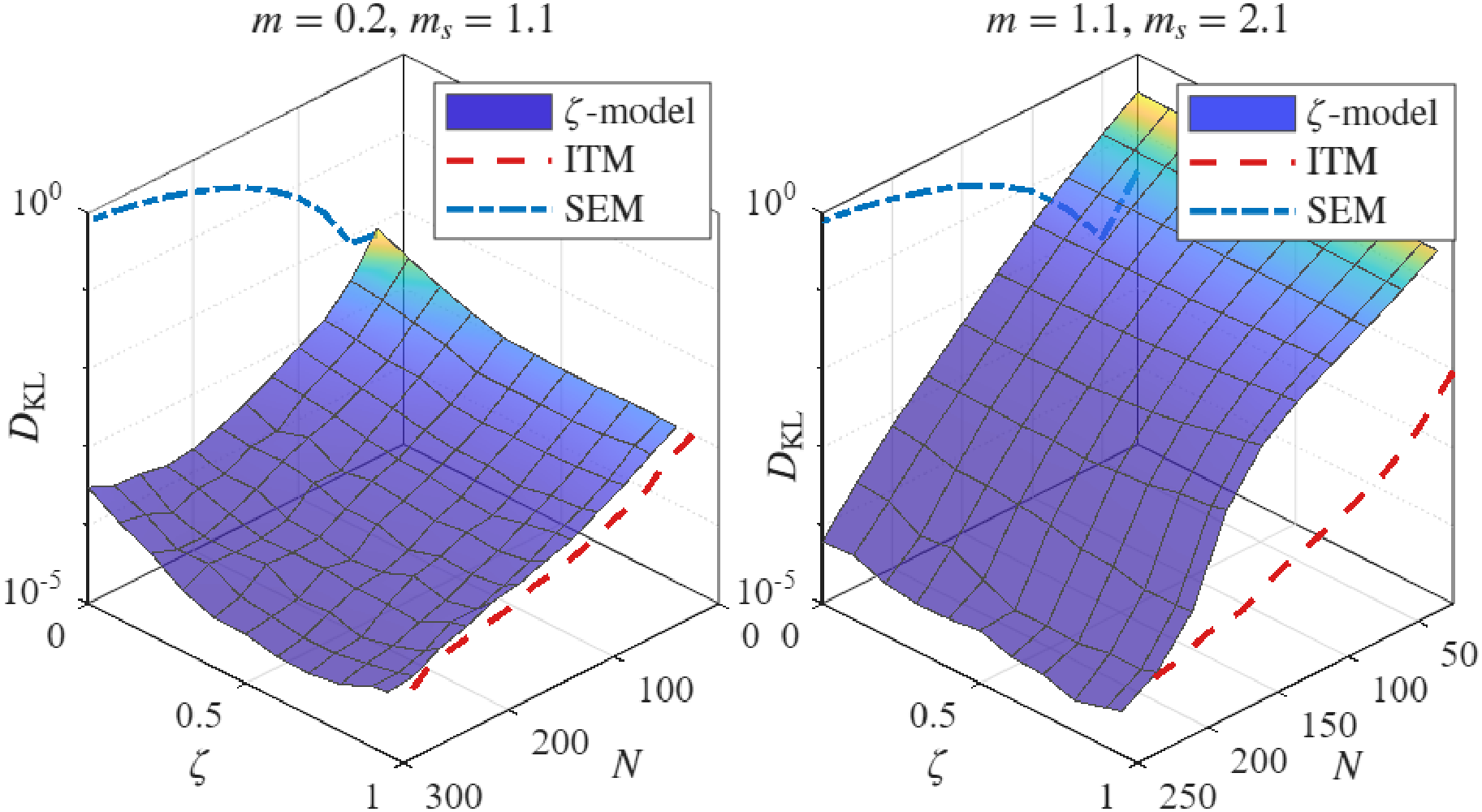}}
	\vspace{-10pt}\caption{$D_{\text{KL}}$ for the $\zeta$-model, ITM, and SEM versus $N$ and $\zeta$ under $1<m_s \leq 2$ and $m_s>2$ with $P_s=0$ dBm.}
	\label{KL_N}\vspace{-12pt}
\end{figure}

\begin{figure}
    \makebox[\columnwidth][c]{%
        \includegraphics[width=.96\columnwidth]{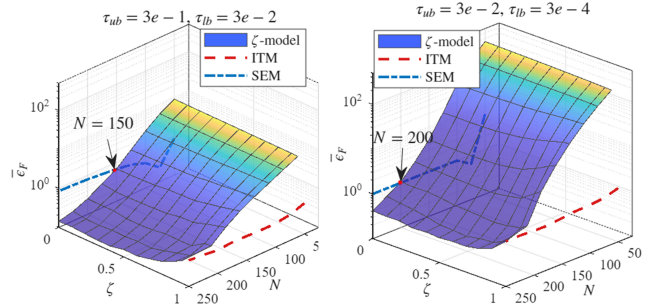}%
    }
    \vspace{-10pt} \caption{$\bar{\epsilon}_F$ for the $\zeta$-model, ITM and SEM versus $N$ and $\zeta$ under $m=1.1, m_s=2.1$ with different $\tau_\text{ub}$ and $\tau_\text{lb}$.}
	\label{MRE_N}\vspace{-18pt}
\end{figure}

Figure \ref{MRE_N} depicts the MRE versus $N$ and $\zeta$ for $\zeta$-model, ITM and SEM under $m=1.1, m_s=2.1$ with $\tau_\text{ub}=1e-2,\tau_\text{lb}=1e-3$ and $\tau_\text{ub}=1e-3, \tau_\text{lb}=1e-4$, respectively. 
As observed, the MRE increases markedly as the evaluation window shifts toward smaller values, indicating that the adopted approximation becomes increasingly sensitive when the target region moves deeper into the left tail. Therefore, additional caution is required when the model is applied to reliability evaluation in deep tail regimes. Moreover, the SEM curve intersects with the $\zeta$-model curve at $\zeta=0$ in both subfigures, revealing an evaluation-window-dependent accuracy boundary. 
As the evaluation window moves further into the left tail, this boundary shifts to a larger $N$, suggesting that more RIS elements are needed before the $\zeta$-model overtakes the SEM in relative-error accuracy under stringent tail-evaluation regimes.



\section{Performance Characterization and Analysis}
In this section, we perform a comprehensive performance analysis for the proposed framework by deriving the AC, OP, and EC. 
Specifically,  we first obtain closed-form expressions for the AC of the ITM and SEM, followed by theoretical upper bounds for the unified $\zeta$-Model and boundary cases. Subsequently, we present the theoretical and asymptotic expressions of the OP and EC for the $\zeta$-Model and boundary cases. 
\textcolor{black}{It is noteworthy that these three metrics are selected to provide a multi-dimensional evaluation of the system, in which AC quantifies the long-term ergodic throughput limit for delay-tolerant services, EC measures the maximum achievable rate under stringent QoS delay constraints for real-time applications, and OP characterizes the fundamental link reliability and availability.}

\subsection{Average Capacity}
The AC quantifies the maximum achievable data rate under the assumption of ergodic channel conditions. Mathematically, it is expressed as follows:\vspace{-5pt} 
\begin{equation}\label{AC}
	\begin{aligned}
		\mathrm{AC}=\mathbb{E}[\log_2(1+\gamma)] &= \int_{0}^{\infty} \log_2(1+\gamma)f_{\gamma}(\gamma)d\gamma,
		\end{aligned}
\end{equation}
which is governed by the global distribution. However, it is difficult to obtain the closed-form expression of the AC under the unified $\zeta$-Model through direct integration. To tackle this issue, we first derive closed-form expressions for the tractable boundary cases ITM and SEM, and then provide an upper bound applicable to the entire $\zeta$-domain.

\subsubsection{Analysis for Boundary Cases}
For the ITM with $m_s > 2$, the AC of ITM, denoted by $\mathrm{AC}_{\gamma_1}$, can be evaluated by substituting Eq.~\eqref{pdfgamma1} into Eq.~\eqref{AC} as follows:
\begin{equation}\label{ACGAMMA1}
	\begin{aligned}
		\mathrm{AC}_{\gamma_1} =
		&\Lambda_2 \int_{0}^{\infty}{\log (1+\gamma)\gamma^{k_{z_2}-1}\exp\left(-\frac{\gamma}{\bar{\gamma}\theta_{z_2}}\right)}d\gamma, \vspace{-3pt} 
	\end{aligned}
\end{equation}
where $\Lambda_2 = {1}/\left({\Gamma(k_{z_2})(\bar{\gamma}\theta_{z_2})^{k_{z_2}}\log2}\right)$. 
With the aid of \cite[Eqs.~({07.34.03.0228.01})(07.34.03.0456.01)]{walfram}, Eq.~\eqref{ACGAMMA1} can be further calculated as follows:
\begin{equation}
	\begin{aligned}
		\mathrm{AC}_{\gamma_1}&= \Lambda_2 \int_{0}^{\infty}{ G_{2,2}^{1,2}\left(\gamma \middle|\, \begin{gathered}
				k_{z_2},k_{z_2}\\
				k_{z_2},k_{z_2}-1
			\end{gathered}\right) G_{0,1}^{1,0}\left(\frac{\gamma}{\bar{\gamma}\theta_{z_2}} \Bigg|  0 \right)}d\gamma\\
		&\overset{(a)}{=}\frac{1}{\Gamma(k_{z_2})\log 2}G_{2,3}^{3,1}\left(\frac{1}{\bar{\gamma}\theta_{z_2}} \middle|\, \begin{gathered}
			0,1\\
			0,0,k_{z_2}
		\end{gathered}\right), ~~~~~m_s>2.
	\end{aligned}
\end{equation}
where $G_{p, q}^{m, n}\left( x\Bigg| \\\begin{gathered}
	{a_{1}, \cdots, a_{p}}\\{b_{1}, \cdots, b_{q}}
\end{gathered}\right)$ is the Meijer's G-function \cite[Eq.~(9.301)]{2007Table} and    
$(a)$ results from \cite[Eqs.~(7.811.1)(9.31.5)]{2007Table}. 

For the ITM with $1<m_s\le 2$, by substituting the PDF in Eq.~\eqref{pdfgamma1} into Eq.~\eqref{AC} and introducing the change of variable $\gamma^\prime=\sqrt{\gamma/\bar{\gamma}}/\theta_{z_1}$ and the identity \cite[Eq.~(07.34.03.0456.01)]{walfram}, the AC can be written as follows:
\begin{equation}
\mathrm{AC}_{\gamma_{1}} = \Lambda_1 \int_0^\infty
{\gamma'}^{k_{z_1}-1}e^{-\gamma'}
G_{2,2}^{1,2}\left(
\bar{\gamma}\theta_{z_1}^2{\gamma'}^2
\;\middle|\;
\begin{matrix}
1,1\\
1,0
\end{matrix}
\right)
d\gamma'.
\label{eq:AC_gamma1_integral}
\end{equation} 
where $\Lambda_1 = {1}/({\Gamma (k_{z_{1}} ) \log 2})$. 
Then, by invoking the Mellin-Barnes representation of the Meijer's $G$-function, the order of integration can be exchanged, yielding
\begin{equation}
	\begin{aligned}
\mathrm{AC}_{\gamma_{1}}
&=
\frac{\Lambda_1}{2\pi i}
\int_{\mathcal{L}}
\frac{\Gamma(1-s)\Gamma^2(s)}{\Gamma(1+s)}
(\bar{\gamma}\theta_{z_1}^2)^s\\ 
&~~~~~~~~~~~~~~~~~~~\times\left[
\int_0^\infty
{\gamma'}^{k_{z_1}+2s-1}e^{-\gamma'}d\gamma'
\right]ds.\vspace{-5pt} 
\end{aligned}\label{eq:AC_gamma1_integral2}
\end{equation}
The inner integral in Eq.~\eqref{eq:AC_gamma1_integral2} is the standard gamma integral $
\Gamma(k_{z_1}+2s)$, which can be further expressed as $ {2^{k_{z_1} + 2s - 1}} \Gamma\left(\frac{k_{z_1}}{2} + s\right) \Gamma\left(\frac{k_{z_1} + 1}{2} + s\right)/{\sqrt{\pi}}$ with the aid of the Legendre duplication formula\cite[Eq.~(8.335.1)]{2007Table}.
By recasting the resulting contour integral into the standard Meijer's $G$-function form, we finally obtain\vspace{-3pt} 
\begin{equation}\label{ACGAMMA1_msleq2_final}
\begin{aligned}
\mathrm{AC}_{\gamma_1}
&=
\frac{2^{k_{z_1}-1}}
{\sqrt{\pi}\,\Gamma(k_{z_1})\log 2}
G_{4,2}^{1,4}\!\left(\!
4\bar{\gamma}\theta_{z_1}^{2}
\;\!\middle|\!\;
\begin{matrix}
1,1,\frac{1-k_{z_1}}{2},1-\frac{k_{z_1}}{2}\\
1,0
\end{matrix}
\right)\!,\\&\quad\quad\quad\quad\quad\quad\quad\quad~~~~~~~~~~~~~~~~~~~~~~~~ {1<m_s\le 2}.
	\end{aligned}
\end{equation}



\vspace{-5pt}The AC of SEM, denoted by $\mathrm{AC}_{\gamma_0^s}$, can be derived following a derivation process similar to the ITM as 
\begin{equation}
		\mathrm{AC}_{\gamma_0^s} \!=\!
		\begin{cases}\!
		\begin{aligned}
		&\frac{1}{\Gamma(k_{x_2})\log 2}G_{2,3}^{3,1}\left(\frac{1}{\bar{\gamma}_{_\text{ef}}\theta_{x_2}} \middle|\, \begin{gathered}
			0,1\\
			0,0,k_{x_2}
		\end{gathered}\right), m_s >2;\\
		&\frac{2^{k_{x_1}-1}}
{\sqrt{\pi}\,\Gamma(k_{x_1})\log 2}\times\\ &
 G_{4,2}^{1,4}\!\left(\!
4\bar{\gamma}_{_\text{ef}}\theta_{x_1}^{2}
\;\!\middle|\;\!
\begin{matrix}
1,1,\frac{1-k_{x_1}}{2},1-\frac{k_{x_1}}{2}\\
1,0
\end{matrix}
\!\right)\!, {1<m_s\le 2}.
	\end{aligned}
	\end{cases}
\end{equation}
The gap between the ITM and SEM quantifies the sensitivity of the long-term average-rate prediction to residual-error exploitability. 
A large gap implies that post-acquisition residual exploitation can materially change the predicted average throughput and therefore needs to be taken seriously.

\subsubsection{Upper Bound Analysis for the Unified $\zeta$-Model}

The upper bound of AC, denoted by $AC^{\text{ub}}$, can be calculated using Jensen's inequality as follows:
\begin{equation}\label{Cub}
	\mathrm{AC}\le \mathrm{AC}^{\text{ub}}=\log_2(1+\mathbb{E}[\gamma]).
\end{equation}
For the $\zeta$-Model, by leveraging the independence between $U$ and $\varepsilon_\zeta$ shown in Section II-C, $\mathbb{E}[\gamma_\zeta]$ can be expressed as
\begin{equation}\label{Ezeta}
\mathbb{E}[\gamma_\zeta]=\mathbb{E}\left[\frac{\bar{\gamma}|U|^2}{\bar{\gamma}|\varepsilon_\zeta|^2+1}\right]=\mathbb{E}[|U|^2]\mathbb{E}\left[\frac{1}{|\varepsilon_\zeta|^2+\frac{1}{\bar{\gamma}}}\right].
\end{equation} 
Then, plugging Eq.\eqref{PDF_varepsilon} into Eq.~\eqref{Ezeta},
 $\mathbb{E}[\gamma_\zeta]$ can be calculated as follows:\vspace{-10pt} 
\begin{equation}\label{Ezeta2}
\begin{aligned}
	\mathbb{E}[\gamma_\zeta]&=\mathbb{E}[|U|^2]\int_{0}^\infty \frac{1}{\varepsilon+\frac{1}{\bar{\gamma}}}\frac{\exp\left(-\frac{\varepsilon}{\sigma_{\varepsilon_\zeta}^2}\right)}{\sigma_{\varepsilon_\zeta}^2} d\varepsilon\\
	&\overset{(a)}=\frac{\mathbb{E}[|U|^2]}{\sigma_{\varepsilon_\zeta}^2}\exp\left(\frac{1}{\sigma_{\varepsilon_\zeta}^2\bar{\gamma}}\right)\int_{\frac{1}{\sigma_{\varepsilon_\zeta}^2\bar{\gamma}}}^\infty \frac{\exp\left(-\varepsilon\right)}{\varepsilon} d\varepsilon\\
	&\overset{(b)}=-\frac{\mathbb{E}[|U|^2]}{\sigma_{\varepsilon_\zeta}^2}\exp(\frac{1}{\sigma_{\varepsilon_\zeta}^2\bar{\gamma}})\mathrm{Ei}(-\frac{1}{\sigma_{\varepsilon_\zeta}^2\bar{\gamma}}),\vspace{-5pt}
\end{aligned}
\end{equation}
where $\mathbb{E}[|U|^2]$ is provided from Appendix \ref{AppendixA} as mentioned above, $(a)$ results from variable substitution, and $(b)$ is obtained by utilizing the one-argument exponential integral function $\mathrm{Ei}(x)$
\cite[Eq.~(8.211.1)]{2007Table}. Then, by substituting Eq. \eqref{Ezeta2} into Eq. \eqref{Cub}, $AC^{\text{ub}}_\zeta$ can be obtained.

For the ITM and SEM, the expectation of $\gamma_1$ and $\gamma_0^s$ can be directly calculated according to Appendix \ref{AppendixA} without the special function $\mathrm{Ei}(x)$, respectively, as follows:\vspace{-2pt}
\begin{equation}\label{Egamma_ITM_SEM}
	\mathbb{E}[\gamma_1] = \bar{\gamma}\mathbb{E}[|Z|^2],~~~~\mathbb{E}[\gamma_0^s] = \bar{\gamma}_{_\text{ef}}\mathbb{E}[|X|^2].\vspace{-3pt}
\end{equation}
For the EBM, Eq.~\eqref{Ezeta2} will reduce to $\mathbb{E}[\gamma_0]$ by setting $\zeta=0$ and replacing $\sigma_{\varepsilon_\zeta}^2 $ with $\sigma_E^2$ and $|U|^2$ with $|X|^2$. By substituting the expectations into Eq. \eqref{Cub}, $AC^{\text{ub}}$ for the ITM, EBM, and SEM can be derived. 
The upper-bound expressions place the ITM, EBM, and SEM, and intermediate-$\zeta$ model on the same analytical scale, which makes their AC gaps directly comparable. This provides a convenient way to assess how conservative the engineering baseline is for throughput-oriented system evaluation.

\vspace{-5pt}
\subsection{Outage Probability}
The OP, denoted by $\mathbb{P}$, is defined as the probability that the instantaneous SNR falls below a predetermined outage threshold $\gamma_{\text{th}}$. 
The OP for the unified $\zeta$-Model, denoted by $\mathbb{P}_{\gamma_\zeta}$, can be directly expressed as follows:
\vspace{-5pt}\begin{equation}
\mathbb{P}_{\gamma_\zeta}=\text{Pr}\left[\gamma_\zeta<\gamma_{\text{th}}\right]=F_{\gamma_\zeta}\left(\gamma_{\text{th}}\right),
\end{equation}
which applies to arbitrary $\zeta\in[0,1]$. Nevertheless, for reliability-oriented engineering design, we place particular emphasis on the $\zeta=0$ boundary, since it corresponds to the practical engineering baseline in which the residual CSI uncertainty is entirely treated as detrimental impairment at the receiver. Therefore, by using the derived CDFs of the EBM and SEM, the corresponding OPs, denoted by $\mathbb{P}_{\gamma_0}$ and $\mathbb{P}_{\gamma_0^s}$, respectively, can be directly obtained as
\vspace{-2pt}
\begin{equation}\label{op}
	\mathbb{P}_{\gamma_0} = F_{\gamma_0}(\gamma_{\text{th}}),~~~~\mathbb{P}_{\gamma_0^s}=F_{\gamma_0^s}(\gamma_{\text{th}}). \vspace{-2pt}
\end{equation}
Different from the AC depending on the entire distribution, the OP is fully determined by the left-tail CDF at the outage threshold. Therefore, accurate left-tail characterization is particularly critical for OP analysis.
To further reveal the reliability limit under CSI imperfections, we next investigate the high-SNR asymptotic behaviors.

\begin{proposition}\label{op_proposition}
In this proposition, we consider \(0 \leq \zeta < 1\), for which the unexplained residual component has a nonzero variance. The boundary case \(\zeta=1\) reduces to the ITM and does not exhibit the same impairment-limited nonzero outage floor. As $\bar{\gamma}\to\infty$, the instantaneous SNR of the unified $\zeta$-Model satisfies $\mathbb{P}_{\gamma_\zeta}^{\text{asy}}=\text{Pr}\left[|U|^2/{|\varepsilon_\zeta|^2}<\gamma_{\text{th}}\right]$. For $m_s > 2$, by plugging Eqs.~\eqref{CDF_U2} and \eqref{PDF_varepsilon} and using \cite[Eq. (6.451.1)]{2007Table}, the asymptotic OP can be derived as follows:\vspace{-5pt}
		\begin{equation}\label{OP_asy_gt2}  \mathbb{P}_{\gamma_\zeta}^{\text{asy}}=\left( 1 + \frac{\theta_{u_2}}{\sigma_{\varepsilon_\zeta}^2 \gamma_{\text{th}}} \right)^{-k_{u_2}},
		\quad m_s>2.
	\end{equation}
	For $1<m_s \leq 2$, the asymptotic OP can be rewritten as
\begin{equation}
	\begin{aligned}\label{OP_asy_le2}
&\mathbb{P}_{\gamma_\zeta}^{\text{asy}}=\text{Pr}\left[|\varepsilon_\zeta|>\frac{|U|}{\sqrt{\gamma_\text{th}}}\right]=
\int_{0}^\infty \widetilde{F}_{|\varepsilon_\zeta|}\left(\frac{u}{\sqrt{\gamma_{\text{th}}}}\right)f_{|U|}(u) du,\\
&\overset{(a)}{=}\frac{1}{\Gamma(k_{u_1})\theta_{u_1}^{k_{u_1}}}
\int_0^\infty
u^{k_{u_1}-1}
\exp\!\left(
-\frac{u}{\theta_{u_1}}
-\frac{u^2}{\gamma_{\text{th}}\sigma_{\varepsilon_\zeta}^2}
\right)du\\
&\overset{(b)}{=}\left(
\frac{\gamma_{\text{th}}\sigma_{\varepsilon_\zeta}^2}{4\theta_{u_1}^2}
\right)^{\frac{k_{u_1}}{2}}
U\!\left(
\frac{k_{u_1}}{2},
\frac{1}{2},
\frac{\gamma_{\text{th}}\sigma_{\varepsilon_\zeta}^2}{4\theta_{u_1}^2}
\right),
\qquad 1<m_s\le 2,\vspace{-3pt}
	\end{aligned}
\end{equation}
where $U(a,b,z)$ denotes the confluent hypergeometric Kummer U function\cite[Eq.~(9.211.4)]{2007Table}, $(a)$ is obtained since $|\varepsilon_\zeta|$ follows Rayleigh distribution, and $(b)$ results from \cite[Eq.~(3.462.1)]{2007Table} and \cite[Eq.~(07.41.26.0007.01)]{walfram}. 
For the EBM, the asymptotic OP follows directly from Proposition~\ref{op_proposition} by replacing $\sigma_{\varepsilon_\zeta}^2$ with $\sigma_E^2$ and $\{k_{u_2},\theta_{u_2}\}$ or $\{k_{u_1},\theta_{u_1}\}$ with $\{k_{x_2},\theta_{x_2}\}$ or $\{k_{x_1},\theta_{x_1}\}$, respectively.

The diversity order of the unified $\zeta$-Model is defined as $G_d
= -\displaystyle\lim_{\bar{\gamma}\to\infty} 
{\log \mathbb{P}^{\text{asy}}_{\gamma_\zeta}(\bar{\gamma})}/{\log \bar{\gamma}}$. 
Since the asymptotic OP in both Eqs.~\eqref{OP_asy_gt2} and \eqref{OP_asy_le2} converges to a non-zero constant independent of $\bar{\gamma}$ for \(0\leq\zeta<1\), the diversity order is
$G_d=0$.
\end{proposition}
\begin{remark}
The above result indicates a transition from a power-limited regime to an impairment-limited regime. In the former, increasing $\bar{\gamma}$ remains beneficial, whereas in the latter, the outage floor is fundamentally determined by the non-exploitable residual error and moment-matching parameters. Hence, once the system approaches this floor, improving channel prediction, residual error exploitation, or adding reflecting paths is more effective than further increasing $\bar{\gamma}$.	
\end{remark}
\vspace{-10pt}



\subsection{Effective Capacity}
To address the diverse QoS requirements in EWC scenarios, we adopt the statistical delay-bounded QoS exponent, $\theta~(\theta>0)$, which represents the exponential decay rate of the delay-bound violation probability. A larger $\theta$ (a faster decay rate) indicates a stringent QoS requirement, while a smaller $\theta$ (a smaller decay rate) implies a more relaxed QoS standard. 

According to \cite{ECCHENG}, the normalized EC of the service process for a given $\theta$, denoted by $\mathrm{EC}$,  can be expressed as follows:
\begin{equation}\label{EC}  
	\begin{aligned}
	\mathrm{EC} & = -\frac{1}{\beta}\log_2 \left(\mathbb{E}_\gamma\left[e^{-\theta R}\right]\right)= -\frac{1}{\beta}\log_2 \left(\mathbb{E}_\gamma\left[\left(1+\gamma\right)^{-\beta} \right]\right),
	\end{aligned}
\end{equation}
where $R=BT\log_2 (1+\gamma)$ denotes the amount of service in bits over a block duration $T$, $B$ is the system spectral bandwidth, and $\beta=BT\theta/\log 2$ represents the normalized QoS exponent. 

While Eq.~\eqref{EC} allows for the numerical evaluation of QoS performance for any $\zeta$, deriving a closed-form solution for the general case is mathematically intractable due to the complex structure of $F_{\gamma_\zeta}$. 
From a robust engineering perspective, to establish the guaranteed performance baseline, we first focus our analytical derivation on the SEM with $\zeta=0$. 
As $m_s > 2$, $\mathbb{E}[(1+\gamma_0^s)^{-\beta}]$ can be directly obtained by substituting $f_{\gamma_0^s}(\gamma)$ in Eq.~\eqref{pdfgamma2s} into Eq.~\eqref{EC} and introducing the change of variable on the basis of $U(a,b,z)$. Therefore, the EC of SEM, denoted by $\mathrm{EC}_{\gamma_0^s}
$, can be expressed as follows:
\vspace{-5pt}
\begin{equation}\label{EC2s_msgt2} 
		\begin{aligned}
	\mathrm{EC}_{\gamma_0^s}=
		&-\frac{1}{\beta} \left[\log_2\left( U(k_{x_2},k_{x_2}+1-\beta,\frac{1}{\bar{\gamma}_{_\text{ef}}\theta_{x_2}})\right)\right.\\
		&~~~~~~~~~~~~~~~~~~~~~~~\left.-k_{x_2}\log_2(\bar{\gamma}_{_\text{ef}}\theta_{x_2})\right], \quad m_s >2.
	\end{aligned}
\end{equation}
In the following, we define $\mathcal{E}_{\gamma}=\mathbb{E}[(1+\gamma)^{-\beta}]$ for simplicity. For $1<m_s \leq 2$, by plugging $f_{\gamma_0^s}(\gamma)$ and applying $t = \sqrt{{\gamma}}/({\bar{\gamma}_{_\text{ef}}}{\theta_{x_1}})$, the expectation term $\mathcal{E}_{\gamma_{0}^s}$ 
can be rewritten as 
\begin{equation}\label{Integral_x}
	\begin{aligned}
    \mathcal{E}_{\gamma_{0}^s} &= \frac{1}{\Gamma(k_{x_1})} \int_0^\infty
t^{k_{x_1}-1}e^{-t}
\left(1+\bar{\gamma}_{_\text{ef}}\theta_{x_1}^2 t^2\right)^{-\beta}
dt.
\end{aligned}
\end{equation}
Next, using the Mellin-Barnes representation\cite[Eqs.~(9.113) and (9.121.1)]{2007Table}, the binomial term $\left(1+\bar{\gamma}_{_\text{ef}}\theta_{x_1}^2 t^2\right)^{-\beta}$ can be rewritten as $
\int_{\mathcal L} \Gamma(s)\Gamma(\beta-s) (\bar{\gamma}_{_\text{ef}}\theta_{x_1}^2)^{-s} ds / {2\pi i}$.
By exchanging the order of integration and applying the nature of the gamma function \cite[ Eq.~(8.310.1)]{2007Table}, Eq.~\eqref{Integral_x} becomes
\begin{equation}
\frac{1}{\Gamma(k_{x_1})\Gamma(\beta)}
\frac{1}{2\pi i}
\int_{\mathcal L} \Gamma(s)\Gamma(\beta-s) \Gamma(k_{x_1}-2s) (\bar{\gamma}_{_\text{ef}}\theta_{x_1}^2)^{-s} ds.
\end{equation}
Then, by applying the Legendre duplication formula\cite[Eq. (8.335.1)]{2007Table} to $\Gamma(k_{x_1}-2s)$ and recasting the resulting contour integral into the standard Meijer's $G$-function form\cite[Eq.~(9.301)]{2007Table}, $\mathrm{EC}_{\gamma_0^s}$ is finally obtained as follows:
\begin{equation}\label{EC2s_msleq2}
	\begin{aligned}
	&\mathrm{EC}_{\gamma_0^s}
=
-\frac{1}{\beta}
\log_2
\left[
\frac{2^{k_{x_1}-1}}
{\sqrt{\pi}\,\Gamma(k_{x_1})\Gamma(\beta)}\right.\\
&\left. \times G_{3,1}^{1,3}\left(
4\bar{\gamma}_{_\text{ef}}\theta_{x_1}^2
\;\middle|\;
\begin{matrix}
1-\beta,\;1-\frac{k_{x_1}}{2},\;\frac{1-k_{x_1}}{2}\\
0
\end{matrix}
\right)
\right], ~ {1<m_s\le 2}.	
	\end{aligned}
\end{equation}
Although the approximate closed-form expressions provide a tractable engineering reference, they are still derived from a simplified alternative model. To capture the high-SNR limit, we derive the
asymptotic EC for the unified $\zeta$-Model.

\begin{proposition}
	For \(0\leq\zeta<1\), the asymptotic expression of the EC for the unified $\zeta$-Model, denoted by $\mathrm{EC}_{{\gamma}_\zeta}^{\text{asy}}$, can be expressed as 
	\begin{equation}\label{EC2ub}  
		\mathrm{EC}_{{\gamma}_\zeta}^{\text{asy}}\!=\!
		\begin{cases}\!
		\begin{aligned}
			&-\frac{1}{\beta} \log_2 \left[\frac{\Gamma(k_{u_2}+1)\Gamma(\beta+1)}{\Gamma(k_{u_2}+\beta+1)}\right.\times\\ & \left.{}_2 F_1\left(\beta,k_{u_2}; k_{u_2}+\beta+1; 1-\frac{\theta_{u_2}}{\sigma_{\varepsilon_\zeta}^2}\right) \right], m_s > 2; \\
			&-\frac{1}{\beta} \log_2 \left[ \frac{2^{k_{u_1}-1}}{\sqrt{\pi}\Gamma(k_{u_1})\Gamma(\beta)}\right.\times\\
			&\! \left.
			G_{3,2}^{2,3}\!\left(\!
			\frac{4\theta_{u_1}^2}{\sigma_{\varepsilon_\zeta}^2}
			\;\!\middle|\;\!
			\begin{matrix}
			1-\beta,\!1-\frac{k_{u_1}}{2},\!\frac{1-k_{u_1}}{2}\\
			0, 1
			\end{matrix}
			\!\right)\!\right]\!, {1<m_s \leq 2}.
		\end{aligned}	
		\end{cases}
	\end{equation}\vspace{-10pt}	
	\begin{IEEEproof} 
		See Appendix \ref{AppendixC}.
	\end{IEEEproof}	
	\end{proposition}
The asymptotic EC indicates that, when a nonzero unexplained residual component remains, the maximum achievable QoS-constrained rate plateaus at high SNR.
\begin{remark}\label{remarkEC}
Although EC depends on the entire distribution of $\gamma$ which is similar to AC, the weighting factor $(1+\gamma)^{-\beta}$ assigns increasingly larger emphasis to the low-SNR region as $\beta$ grows.  
Combined with the observation in Fig.~\ref{CDF} that increasing \(\zeta\) shifts the SNR distribution to the right and amplifies the approximation error in a fixed left-tail regime, we infer that a larger \(\zeta\) can improve the delay-limited performance, but may reduce the analytical accuracy of EC evaluation.
\end{remark} 
Although optimization is beyond the scope of this paper, the derived expressions can be directly used as analytical objective/constraint evaluators, thus avoiding repeated Monte Carlo simulations and supporting future design of UAV deployment, RIS configuration, and power allocation under capacity, outage, and QoS constraints.

\section{Numerical Results and Discussions}
\begin{figure}
{\includegraphics[width=.95\columnwidth]{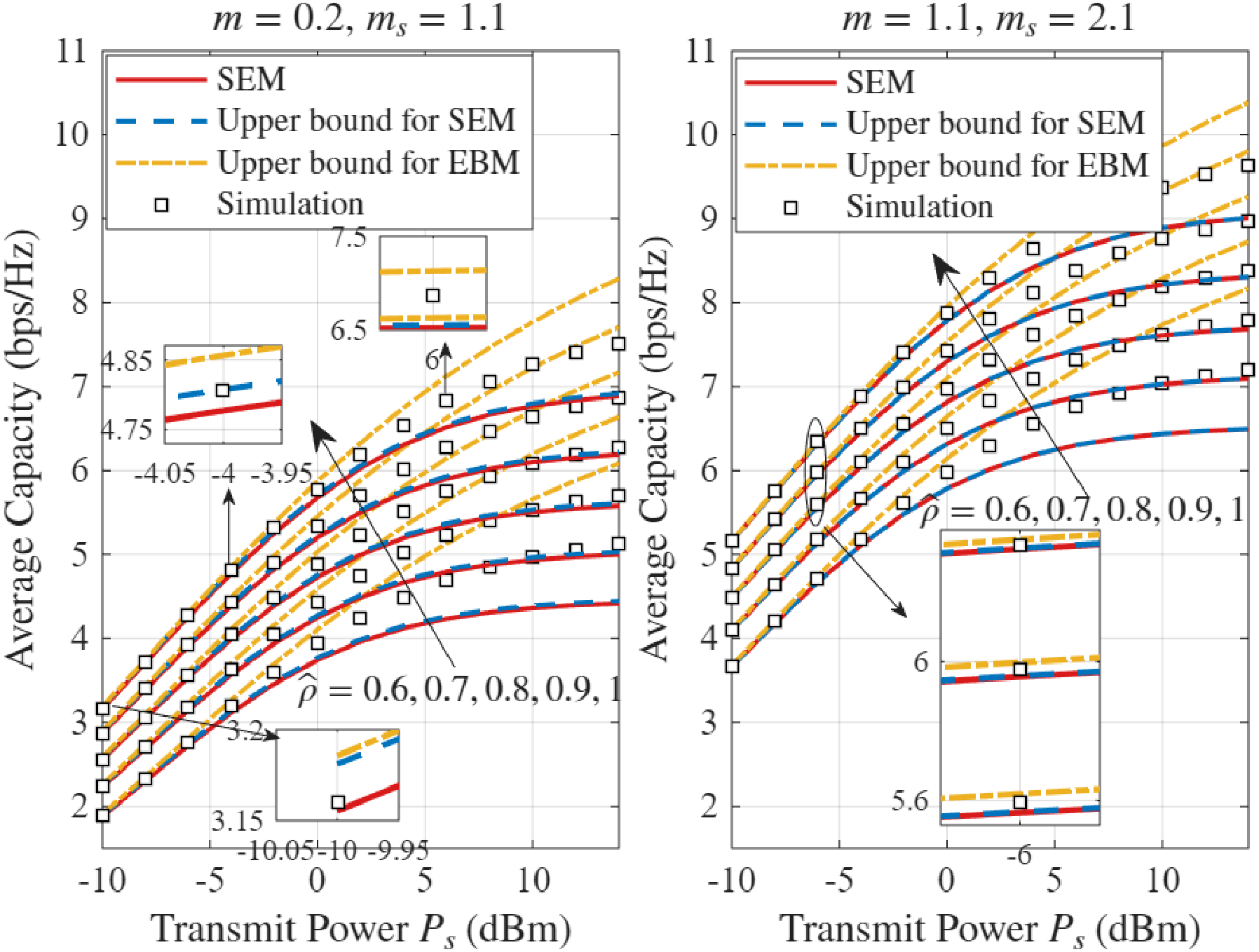}}
	\vspace{-0.2cm}\caption{Average capacity versus $P_s$ corresponding to $\widehat{\rho}$ for the case of $\zeta=0$ under different $m_s$ regimes.}\vspace{-10pt}
	\label{AC_vs_Ps_diff_rho}
\end{figure}
In this section, Monte Carlo simulations are used to validate the developed analytical results. We examine not only the two boundary cases, i.e., $\zeta=0$ and $\zeta=1$, but also the intermediate $\zeta$ values of the unified $\zeta$-Model, in order to verify the validity of the proposed framework and to quantify the performance gap between the boundary benchmarks. 
Unless otherwise stated, two representative fading-shadowing settings are considered: $\left\{ m,m_s\right\}=\left\{0.2,1.1\right\}$ for the regime of $1<m_s\le 2$ and $\left\{m,m_s\right\}=\left\{1.1,2.1\right\}$ for the regime of $m_s>2$, representing severe and moderate fading-shadowing conditions, respectively. This setup confirms that the proposed analysis remains effective across both analytical regimes.

Unless otherwise specified, the main parameters are given as follows. The UE and high-altitude UAV-BS are located at $\mathbf{q_0}=(0,0,0)$ and $\mathbf{q_b}=(300,300,1000)$, respectively. Three low-altitude R-UAVs hovering at the height of $100$ m are uniformly distributed on a $150$ m radius circle centered at the UE. The path loss parameters are set as 
$\alpha=2$, $a=4.88$, $b=0.4472$, $\eta_{\text{LoS}}=0.1$ dB, and $\eta_{\text{NLoS}}=20$ dB. The wavelength is $\lambda=$0.15 m. The communication system operates with SNR threshold  $\gamma_\text{th}=$ 10 dB and noise power $N_0$= -94 dBm. The values of $B$ and $T$ are set to 20 MHz and 2 ms, respectively. The number of quantization bits of each R-UAV is set to $q_k=3, \forall k$. 
For channels with outdated and estimation error, we set $\widehat{\rho}=\widehat{\rho}_d=\widehat{\rho}_{k,n}=\widetilde{\rho}=\widetilde{\rho}_d=\widetilde{\rho}_{k,n}=0.8, \forall k,n$. For channel power normalization, we set $\Omega^d=\Omega^{k,n}=\widehat{\sigma}^2=\widehat{\sigma}_d^2=\widehat{\sigma}_{k,n}^2=\widetilde{\sigma}^2=\widetilde{\sigma}_d^2=\widetilde{\sigma}_{k,n}^2=1$. 
The truncation orders of the Gauss-Laguerre, Gauss-Hermite, and Gauss-Jacobi quadratures are set as
\(N_q=N_h=N_J=20\). For a given quadrature order, the associated nodes and weights are deterministic and can be obtained from standard quadrature tables or built-in numerical routines.

Figure~\ref{AC_vs_Ps_diff_rho} shows the AC results versus transmit power $P_s$ corresponding to different outdated-CSI correlation coefficients $\widehat{\rho}$ under different channel conditions for the $\zeta=0$ case, including the closed-form expression of the SEM and the upper bounds of the EBM and SEM. In both subfigures, the closed-form SEM curves remain close to the corresponding upper bounds, which confirms the tightness of the bound for the simplified model. 
As $P_s$ or $\widehat{\rho}$ increases, the simulation results first approach the theoretical and upper bound SEM curves, then gradually approach the upper bound EBM curve. Additionally, the gap between all theoretical curves will be reduced and converge to the simulation results with the increase of $\widehat{\rho}$ and $N$. However, all the derived results are unable to accurately fit the simulations at extremely high-SNR regimes.
Therefore, these results indicate that the theoretical SEM and its upper bound can provide a reasonably accurate approximation in low-SNR or high-error regimes, while the upper bound of EBM should be adopted in moderate-SNR or low-error scenarios.

\begin{figure}
	\centering
	{\includegraphics[width=0.91\columnwidth]{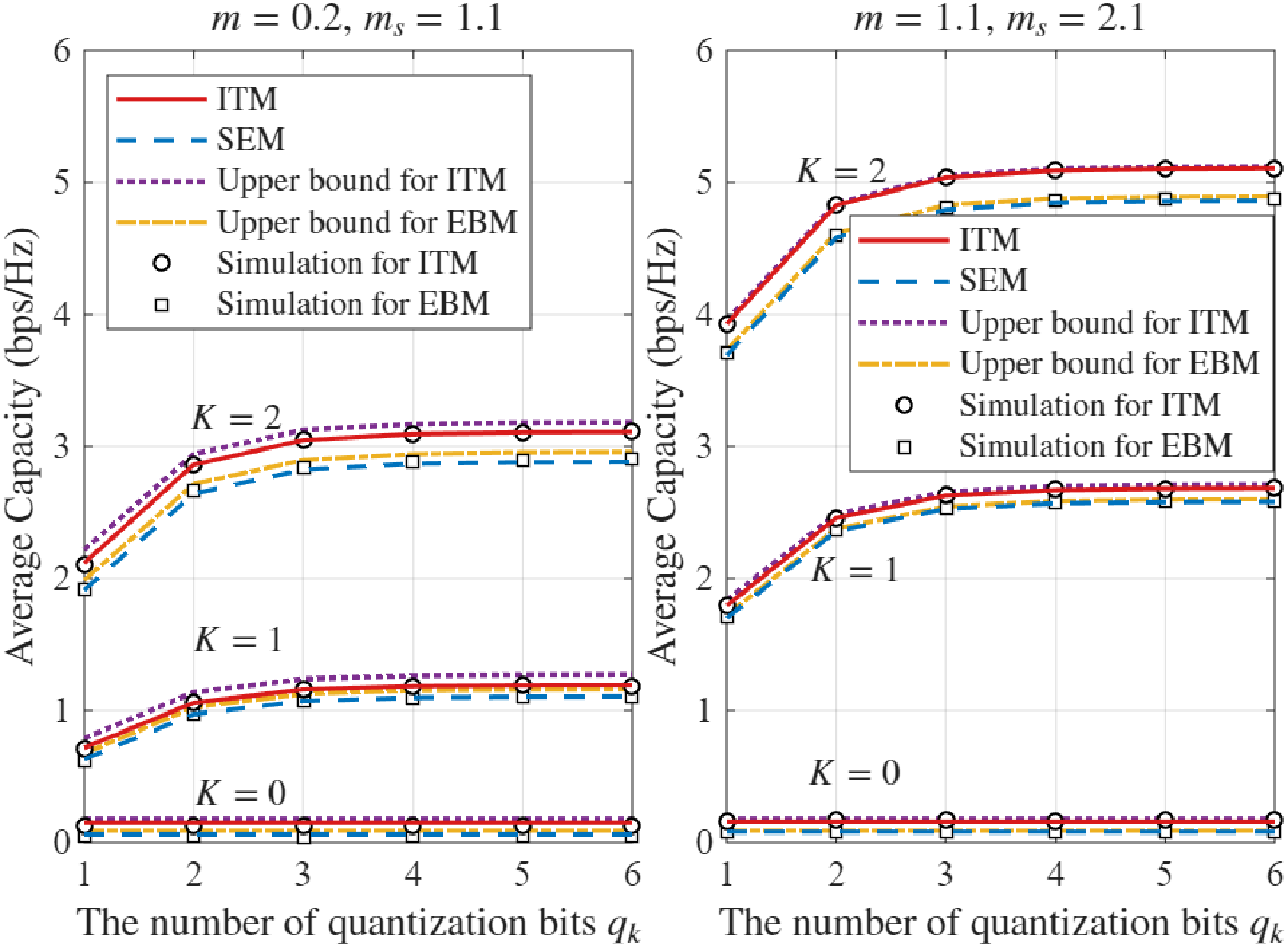}}
	\vspace{-5pt}\caption{Average capacity versus $q_k$ corresponding to different $K$ with $\zeta=0$ and $\zeta=1$.}\vspace{-15pt}
	\label{AC_MODEL1_2_qtp}
\end{figure}


Figure~\ref{AC_MODEL1_2_qtp} shows the AC versus the RIS quantization resolution $q_k$ for different numbers of R-UAVs $K$, where the closed-form expressions of ITM and SEM are compared with their corresponding upper bounds and simulations. When $K=0$, all curves become almost insensitive to $q_k$, because no reflected link is available and the RIS quantization resolution becomes irrelevant. 
For $K=1$ and $K=2$, the AC increases rapidly when $q_k$ grows from 1 to about 3 bits, and then gradually saturates, showing that a moderate quantization resolution is sufficient to capture most of the potential performance benefits and offering an attractive trade-off between system performance and the hardware complexity of the RIS. 
It is also observed that increasing $K$ significantly improves the AC, and the gain becomes more pronounced under milder fading/shadowing conditions. Additionally, the gap between the ITM and SEM enlarges with $K$, indicating that stronger reflected components bring larger gains caused by receiver-side exploitation of the residual error. 

\begin{figure}
	\centering
	{\includegraphics[width=0.93\columnwidth]{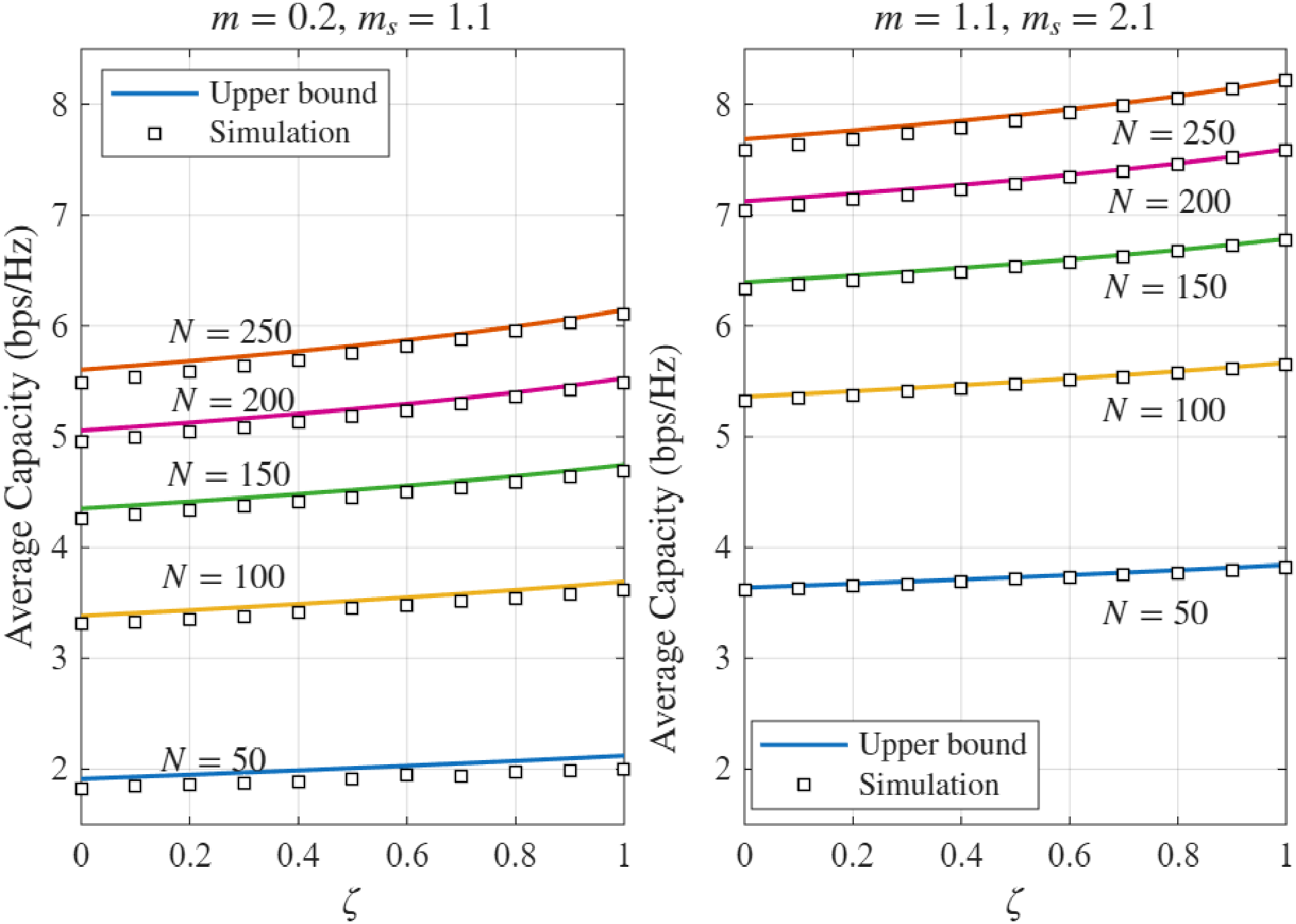}}
	\vspace{-5pt}\caption{{Average capacity versus $\zeta$ corresponding to different $N$ for $\zeta$-model.}}\vspace{-5pt}
	\label{AC_MODEL1_2}
\end{figure}

Figure~\ref{AC_MODEL1_2} depicts the AC versus $\zeta$ for different numbers of RIS elements $N$, where the upper bound and the simulation results of the unified $\zeta$-Model are compared. For the considered parameter range, the AC increases monotonically with $\zeta$ for all considered $N$, and the upper-bound curves remain close to the simulation results over the entire $\zeta$ range. 
Additionally, with an increase of $N$, AC continues to increase and the performance gain from increasing $\zeta$ becomes more visible. This suggests that the benefit of receiver-side residual-error exploitation is amplified when the reflected links are stronger. 

\begin{figure}
	\centering
	\includegraphics[width=1\columnwidth]{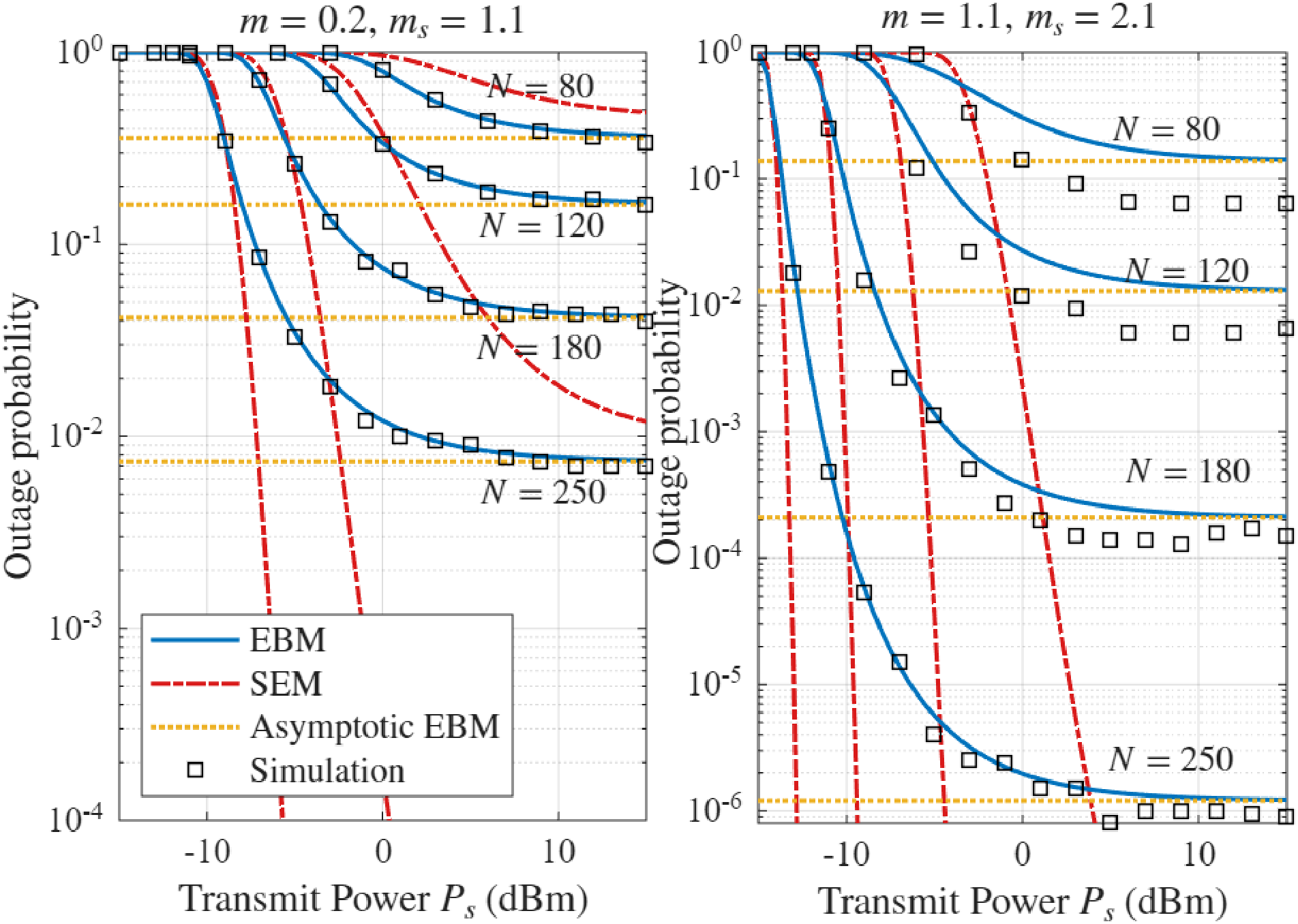}
	\vspace{-10pt}\caption{{Outage probability versus $P_s$ corresponding to different $N$ for EBM and SEM.}}\vspace{-10pt}
	\label{OP_Ps_N}
\end{figure}

Figure~\ref{OP_Ps_N} presents the OP versus $P_s$ with different $N$ for $\zeta=0$, comparing the closed-form OPs of EBM and SEM, and the asymptotic OP of EBM against the simulations. In both subfigures, the OP decreases rapidly with $P_s$ in the low and moderate $P_s$ regimes, then approaches a nonzero floor as $P_s$ becomes sufficiently large. Moreover, increasing $N$ consistently shifts the outage curves downward and lowers the corresponding floor. 
It is also observed that both the closed-form and asymptotic EBM curves fit simulations more accurately than the SEM in the left subfigure, while the discrepancy between the EBM-based results and the simulations becomes more visible at high $P_s$ in the right subfigure. The reason is that the lighter fading/shadowing condition leads to smaller OPs under the same $\gamma_{\text{th}}$, so that the corresponding OP evaluation probes a much more extreme left-tail portion of the SNR distribution. Since the OP in this regime is governed mainly by the local tail behavior, even a mild approximation mismatch will be noticeably amplified on the logarithmic scale.
Nevertheless, the proposed approximation remains useful for identifying the performance ordering, the outage floor, and the regime where increasing $P_s$ is no longer the most efficient means of reliability improvement.
The performance gap between the two channel conditions further shows that milder channel conditions lead to faster outage decay and a lower asymptotic floor, especially for larger $N$ cases.

\begin{figure}
	\centering
	{\includegraphics[width=.99\columnwidth]{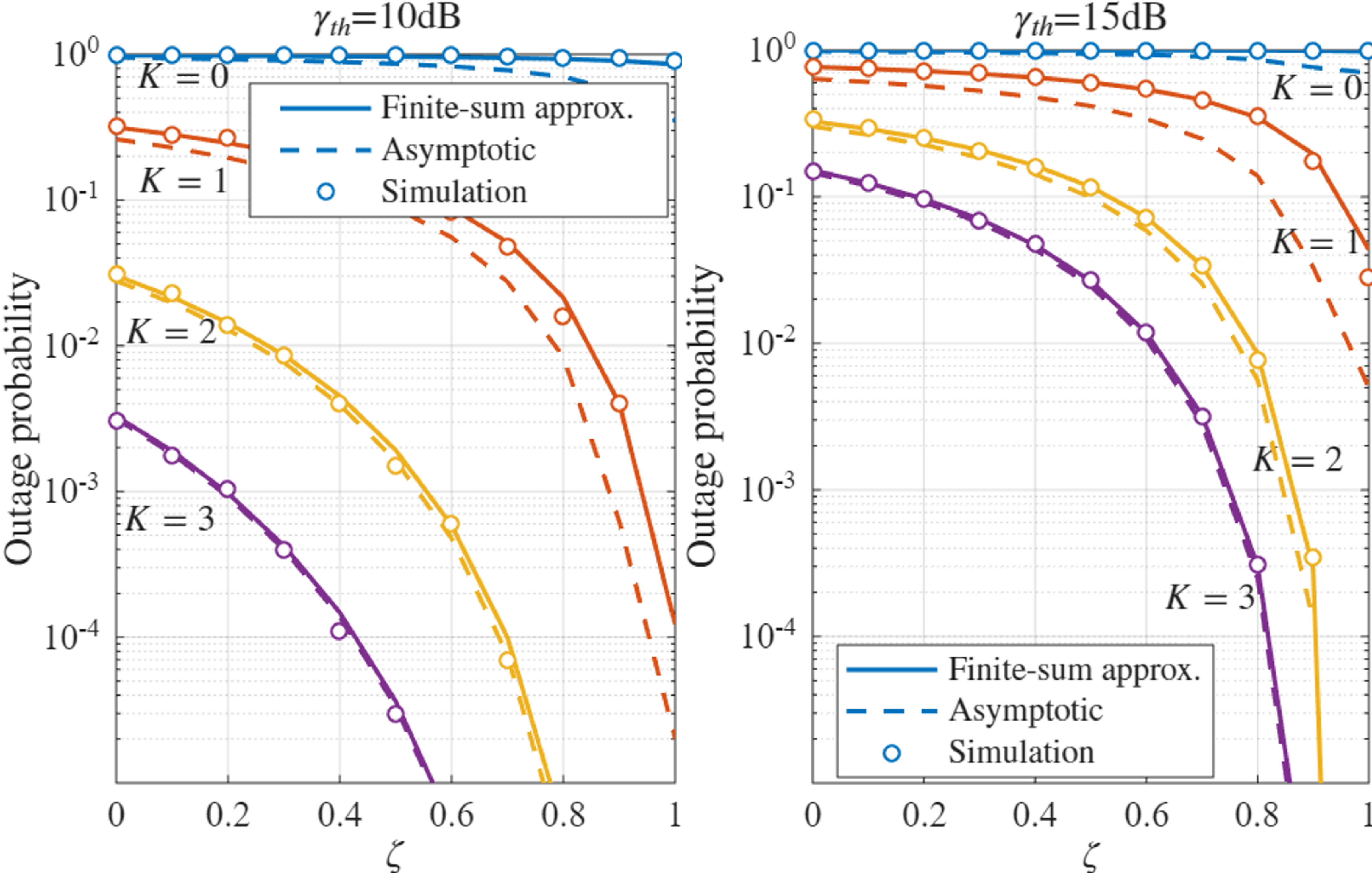}}
	\vspace{-15pt}\caption{Outage Probability versus $\zeta$ corresponding to different $K$ with $m=0.2, m_s=1.1$ and $P_s=15$ dBm under $\gamma_{\text{th}}=10$  dB and $\gamma_{\text{th}} = 15$ dB, respectively.}\vspace{-15pt}
	\label{OP_zeta_K}
\end{figure}
Figure~\ref{OP_zeta_K} shows the finite-sum approximate and asymptotic OPs versus $\zeta$ for different $K$ and outage thresholds under $1<m_s\leq 2$. A common trend in both subfigures is that the OP decreases with $\zeta$, suggesting that improved residual-error exploitability can enhance outage robustness. 
In addition, increasing $K$ brings considerable performance gain compared with the direct-link baseline, indicating that a richer reflected-link structure improves both the effective link quality and the outage robustness. 
Moreover, a similar discrepancy appears when $\gamma_{\text{th}}$ is smaller, the reason of which is essentially similar to that in Fig.~\ref{OP_Ps_N}: a smaller $\gamma_{\text{th}}$ drives the system into a very-low-outage regime. 
Therefore, the proposed model should be interpreted as an efficient tool for trend analysis or parameter screening, while highly accurate prediction in the deep-tail region may require further tail-specific refinement or system parameter adjustment to avoid inaccurate prediction region.

\begin{figure}
	\centering
	{\includegraphics[width=0.95\columnwidth]{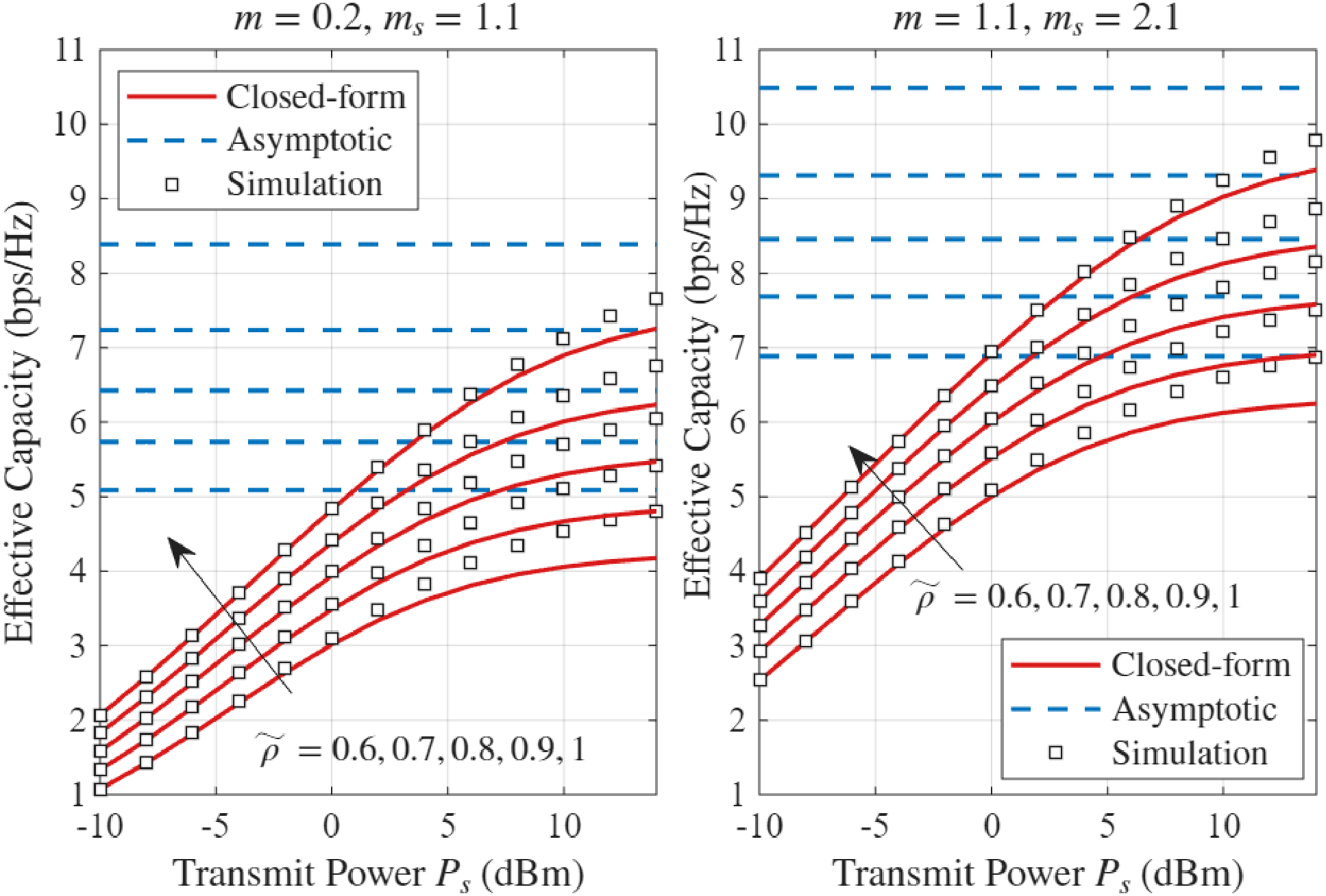}}
	\vspace{-0.2cm}\caption{Effective capacity versus $P_s$ corresponding to different estimated correlation coefficients 
	$\widetilde{\rho}$ with $\beta = 1e^{-3}$ for SEM.}\vspace{-10pt}
	\label{EC_Ps_De}
\end{figure}

Figure~\ref{EC_Ps_De} depicts the EC versus $P_s$ for different estimated correlation coefficients $\widetilde{\rho}$ under $\zeta=0$, including the closed-form EC of the SEM, the asymptotic EC of the EBM, and simulations. 
In both subfigures, the SEM curves match the simulation results well in the low-SNR region, which suggests that the SEM provides an accurate and tractable approximation before the residual CSI error becomes dominant. As $P_s$ increases, the simulation curves gradually deviate from the SEM curves and move toward the horizontal asymptotic EBM curve, revealing an EC ceiling under fixed CSI errors. 
It is also observed that a larger $\widetilde{\rho}$ significantly improves the EC and shifts the asymptotic ceiling upward, indicating that better CSI quality improves performance not only by strengthening the useful term but also by reducing the unexplainable residual impairment. These observations suggest that enhancing channel estimation accuracy yields more remarkable performance gain than merely increasing $P_s$ in high-SNR regimes.

\begin{figure}
	\centering
	{\includegraphics[width=0.93\columnwidth]{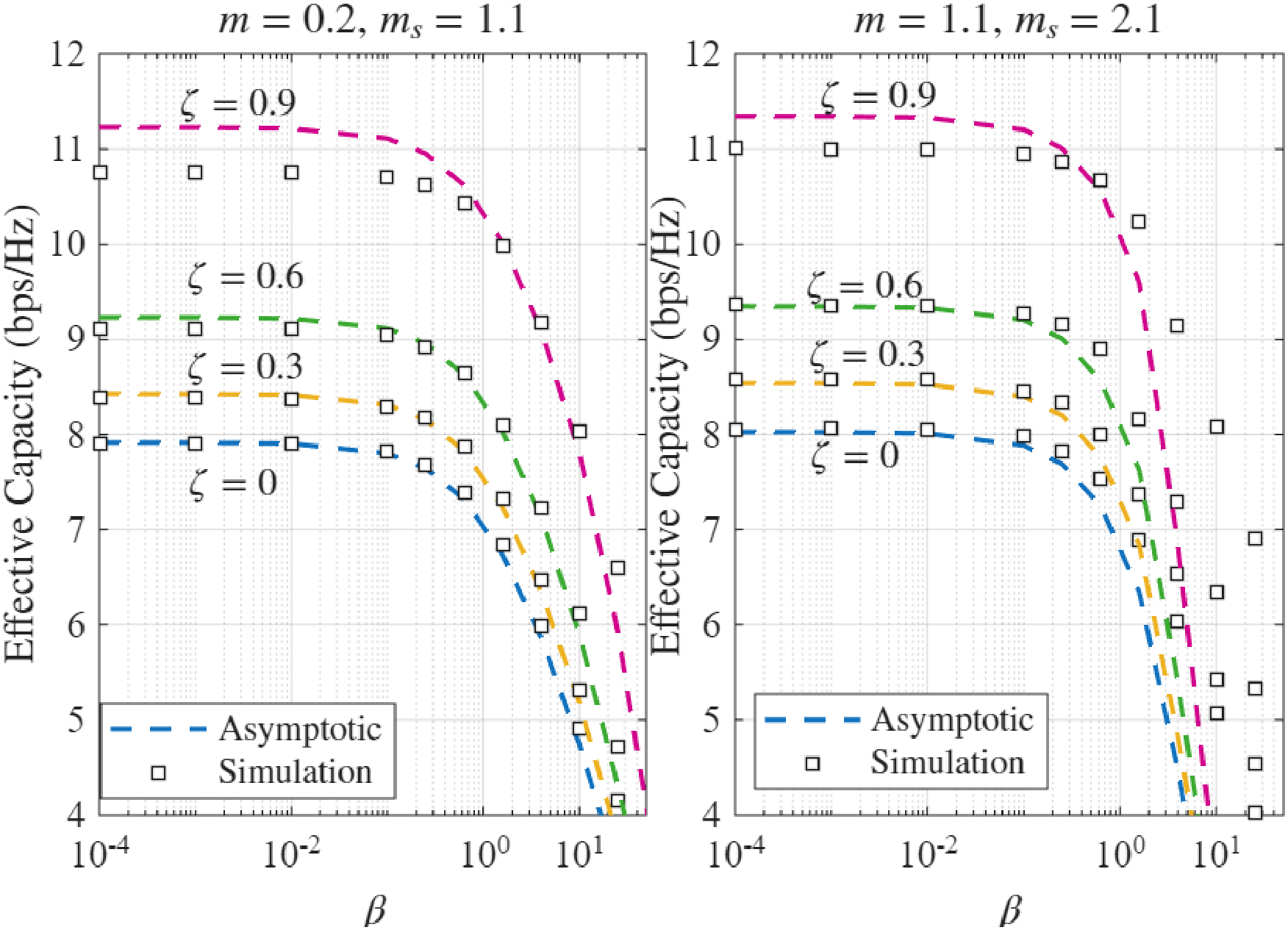}}
	\vspace{-5pt}\caption{Effective capacity versus $\beta$ corresponding to different $\zeta$ for $\zeta$-Model with $P_s  = 15$ dBm.}\vspace{-10pt}
	\label{EC_beta_zeta}
\end{figure}

Figure~\ref{EC_beta_zeta} shows the EC versus the normalized QoS exponent $\beta$ for different $\zeta$ under two fading/shadowing conditions. In both subfigures, the EC increases monotonically with $\zeta$ and decreases monotonically with $\beta$, which is consistent with the fact that a larger $\beta$ imposes a more stringent statistical delay requirement. 
Additionally, the asymptotic results agree well with the simulations in the small/moderate-$\beta$ region, but the mismatch occurs as $\beta$ and $\zeta$ further increase, which is consistent with \textit{Remark} \ref{remarkEC}. 
Moreover, the fitting accuracy with a large $\beta$ of the left subfigure is better than the right subfigure, the reason of which is analogous to previous discussion in Fig.~\ref{OP_Ps_N}. 
Overall, the proposed asymptotic analysis still captures the correct monotonic trends and performance ordering with respect to both $\beta$ and $\zeta$. 
From a design perspective, improving receiver-side residual-error exploitation is an effective approach to enhance delay-constrained throughput, especially when the QoS constraint is not excessively stringent, while the large-$\beta$ region may require more tail-refined evaluation for highly accurate numerical optimization.

\section{Conclusions}
In this paper, we provided a comprehensive analytical framework for emergency heterogeneous UAV systems under composite Fisher-Snedecor $\mathcal{F}$ fading and practical impairments. To resolve the inconsistencies in existing CSI error modeling, a unified SNR framework, termed the $\zeta$-Model, was proposed to reveal the inherent mechanism between different boundaries, including ITM, EBM, and SEM. 
Based on gamma moment-matching, we derived closed-form or finite-sum approximate expressions, together with asymptotic characterizations for key performance metrics. 
The simulation and analytical results highlighted the importance of jointly optimizing CSI accuracy, RIS quantization resolution, and array size to achieve a balanced trade-off between reliability, complexity, and energy efficiency in realistic EWC deployments.

\begin{appendices}
\section{Derivation of the expectation for $n$-th moment of $|U|$}\label{AppendixA}
\subsubsection{ $m_s >2$ } We start from the case of $m_s>2$ to derive the distribution of $|U|^2$. 
First, we need to calculate the expectation and variance of second and fourth moment of $|U|$. The second moment of $|U|$ can be calculated as
\begin{equation}
    \begin{aligned}
    |U|^2&=|X|^2+|E_\zeta|^2+2\operatorname{Re}(X{E_\zeta^\ast}).
    \end{aligned}\label{U2}
\end{equation}
We adopt the moment-level closure that \(X\) is independent of $E_\zeta$ to make the Gamma moment-matching approximation tractable. 
Accordingly, $\mathbb{E}[\operatorname{Re}(X{E_\zeta^\ast})]=0$ 
and $\mathbb{E}[|U|^2]$ can be expressed as
\begin{equation}\label{EU2}
    \begin{aligned}
    \mathbb{E}[|U|^2]=\mathbb{E}[|X|^2]+\mathbb{E}[|E_\zeta|^2],
    \end{aligned}
\end{equation}
where $\mathbb{E}[|X|^2]$ can be calculated as
\begin{equation}\label{EX2}
    \begin{aligned}
        &\mathbb{E}[|X|^2]=\mathbb{E}[(D+\operatorname{Re}(A))^2+\operatorname{Im}(A)^2]\\&=\mathbb{E}[D^2]+\mathbb{E}[\operatorname{Re}(A)^2]+\mathbb{E}[\operatorname{Im}(A)^2]+2\mathbb{E}[D]\mathbb{E}[\operatorname{Re}(A)].
    \end{aligned}
\end{equation}
Taking the square of Eq.~\eqref{U2} and then computing its expectation yields the expectation of  $|U|^4$, denoted by $\mathbb{E}[|U|^4]$, as follows:
\begin{equation}
 \begin{aligned}
    \mathbb{E}[|U|^4]&=\mathbb{E}\left[(|X|^2+|E_\zeta|^2+2\operatorname{Re}(X{{E_\zeta^\ast}}))^2\right]\\
    &\overset{(a)}{=}\mathbb{E}[|X|^4]+\mathbb{E}[|E_\zeta|^4]+4\mathbb{E}[|X|^2]\mathbb{E}[|{E_\zeta}|^2],
    \end{aligned}
\end{equation}
where $(a)$ follows from $\mathbb{E}[\operatorname{Re}(X{E_\zeta^\ast})^2]=\mathbb{E}[|X|^2]\mathbb{E}[|{E_\zeta}|^2]/2$. Here, $\mathbb{E}[|X|^4]$ can be calculated as
\begin{equation}\label{EX4}
    \begin{aligned}
&\mathbb{E}[|X|^4]=\mathbb{E}[(\operatorname{Re}(A)+D)^2+\operatorname{Im}(A)^2]^2\\
&=\mathbb{E}[\operatorname{Re}(A)^4]+4\mathbb{E}[\operatorname{Re}(A)^3]\mathbb{E}[D]+6\mathbb{E}[\operatorname{Re}(A)^2]\mathbb{E}[D^2]\\&+4\mathbb{E}[\operatorname{Re}(A)]\mathbb{E}[D^3]+\mathbb{E}[D^4]+2\mathbb{E}[\operatorname{Im}(A)^2]\left(\mathbb{E}[\operatorname{Re}(A)^2]\right.\\&\left.+{2\mathbb{E}[\operatorname{Re}(A)]\mathbb{E}[D]}+\mathbb{E}[D^2]\right)+\mathbb{E}[\operatorname{Im}(A)^4].\end{aligned}
\end{equation}
According to the statistical properties of CSCG RVs, for $x \sim \mathcal{CN}(0, \sigma^2)$, the expectation of the $n$-th order moment of $|x|$ 
is given by
\begin{equation}\label{E|X|}
\mathbb{E}\left[|x|^n\right] =\sigma^n\Gamma(n/2+1).
\end{equation}
Consequently, based on the \emph{Remark 1} and the definition of $\zeta$ in Eq.~\eqref{zeta_def}, 
$\mathbb{E}[|E_\zeta|^2]$ and $\mathbb{E}[|E_\zeta|^4]$ can be directly obtained by substituting $\sigma_E^2$ into Eq.~\eqref{E|X|} as follows:
\begin{equation}
\begin{cases}
\begin{aligned}
&\mathbb{E}[|E|^2] =\sigma_E^2, ~~~~~~~~~~~~~~\mathbb{E}[|E|^4]=\sigma_E^4\Gamma(3)=2\sigma_E^4,\\
&\mathbb{E}[|E_\zeta|^2] =\zeta\sigma_E^2=\sigma_{E_\zeta}^2, ~~~\mathbb{E}[|E_\zeta|^4]=\sigma_{E_\zeta}^4\Gamma(3)=2 \zeta^2\sigma_E^4.
\end{aligned}
\end{cases}\label{mean_nth_order_G_O}  
\end{equation}

Based on \cite[Eqs. (14) and (15)]{tcomF}, for the $i$-th moment of a RV $h$ following Fisher-Snedecor $\mathcal{F}$ distribution with parameters $\left\{m, m_s\right\}$, its expected value can be derived by substituting its PDF into \cite[Eq. (3.194.3)]{2007Table}, yielding
\begin{equation}
    \begin{aligned}
        \mathbb{E}[|h|^i] & =\left(\frac{m_{s}\Omega_{m}}{m\Omega_{s}}\right)^{\frac{i}{2}}\frac{B\left(m+\frac{i}{2},m_{s}-\frac{i}{2}\right)}{B(m,m_{s})}.
    \end{aligned}\label{Eh}
\end{equation} 
{Due to the convergence condition of beta function, $m_s > 2$ is required to exist as an actual integral for $\mathbb{E}[|h|^4]$.} 
By substituting parameter set $\left\{m^d, m^d_s\right\}$ into Eq.~\eqref{Eh}, $\mathbb{E}[|\widetilde{h}_d|^{i}]$ can be directly calculated, yielding the expectation of the $i$-th moment with respect to $D$ as follows:
\begin{equation}
	\mathbb{E}[D^i]=\left(\frac{\widehat{\rho}_{d}\widetilde{\rho}_{d}}{\sqrt{PL_d}}\right)^i\left(\frac{m_{s}^d\Omega_{m}^d}{m^d\Omega_{s}^d}\right)^{\frac{i}{2}}\frac{B\left(m^d+\frac{i}{2},m_{s}^d-\frac{i}{2}\right)}{B(m^d,m_{s}^d)}.
\end{equation}

Since $|\widetilde{\chi}_{k,n}|$ and $\theta_{k,n}$ are mutually independent, we have
$\mathbb{E}[\operatorname{Re}(A)]=\sum_{k}^{K}\sum_{n}^{N}\mathbb{E}[\operatorname{Re}(A_{k,n})]$ and $\mathbb{E}[\operatorname{Im}(A)]=\sum_{k}^{K}\sum_{n}^{N}\mathbb{E}[\operatorname{Im}(A_{k,n})]$, where 
\begin{equation}\label{A_kn}
\begin{cases}
    \mathbb{E}[\operatorname{Re}(A_{k,n})]=\frac{\widehat{\rho}_{k,n}\widetilde{\rho}_{k,n}}{\sqrt{PL_{k,n}}}\mathbb{E}[|\widetilde{\chi}_{k,n}|]\mathbb{E}[\text{cos}(\theta_{k,n})],\\
    \mathbb{E}[\operatorname{Im}(A_{k,n})]=\frac{\widehat{\rho}_{k,n}\widetilde{\rho}_{k,n}}{\sqrt{PL_{k,n}}}\mathbb{E}[|\widetilde{\chi}_{k,n}|]\mathbb{E}[\text{sin}(\theta_{k,n})].\\
    \end{cases}  \vspace{-3pt} 
\end{equation}
Then, $\mathbb{E}[\operatorname{Re}(A_{k,n})^2]$ and $\mathbb{E}[\operatorname{Im}(A_{k,n})^2]$ can be expressed as 
\begin{equation}\label{A_kn2}
\begin{cases}
    \mathbb{E}[\operatorname{Re}(A_{k,n})^2]=\frac{\widehat{\rho}_{k,n}^2\widetilde{\rho}_{k,n}^2}{PL_{k,n}}\mathbb{E}[|\widetilde{\chi}_{k,n}|^2]\mathbb{E}[\text{cos}^2(\theta_{k,n})],\\
    \mathbb{E}[\operatorname{Im}(A_{k,n})^2]=\frac{\widehat{\rho}_{k,n}^2\widetilde{\rho}_{k,n}^2}{PL_{k,n}}\mathbb{E}[|\widetilde{\chi}_{k,n}|^2]\mathbb{E}[\text{sin}^2(\theta_{k,n})].
\end{cases}   \vspace{-3pt}
\end{equation}
Similarly, $\mathbb{E}[|\widetilde{\chi}_{k,n}|^i]$ can be obtained by substituting parameter set $\left\{m^{k,n}, m^{k,n}_s\right\}$ into Eq.~\eqref{Eh}. 
Since $\theta_{k,n}$ follows the uniform distribution, we have
\begin{equation}
\begin{cases}
   \mathbb{E}[\sin(\theta_{k,n})] =0,~\mathbb{E}[\cos(\theta_{k,n})] = \frac{2^{q_k}}{\pi}\sin(\frac{\pi}{2^{q_k}}),\\
   \mathbb{E}[\cos^2(\theta_{k,n})] = \frac{1}{2} + \frac{2^{q_k-2}}{\pi}\sin\left( \frac{\pi}{2^{q_k-1}} \right),\\
   \mathbb{E}[\sin^2(\theta_{k,n})] = \frac{1}{2} - \frac{2^{q_k-2}}{\pi}\sin\left( \frac{\pi}{2^{q_k-1}}\right).\\
\end{cases} \label{theta}  
\end{equation}
Therefore, the expectations and variances of $\operatorname{Re}(A_{k,n})$ and $\operatorname{Im}(A_{k,n})$, i.e., $\mathbb{E}[\operatorname{Re}(A_{k,n})]$, $\mathbb{E}[\operatorname{Im}(A_{k,n})]$, $\sigma_{\operatorname{Re}(A_{k,n})}^2=\mathbb{E}[\operatorname{Re}(A_{k,n})^2]-\mathbb{E}[\operatorname{Re}(A_{k,n})]^2$, and $\sigma_{\operatorname{Im}(A_{k,n})}^2=\mathbb{E}[\operatorname{Im}(A_{k,n})^2]-\mathbb{E}[\operatorname{Im}(A_{k,n})]^2$, can be directly obtained by substituting Eq.~\eqref{theta} into Eqs.~\eqref{A_kn} and \eqref{A_kn2}.

For a large number of RIS elements $N$, the CLT approximation can be applied to approximate  $\operatorname{Re}(A_k)=\sum_{n}^{N}\operatorname{Re}(A_{k,n})$ and $\operatorname{Im}(A_k)=\sum_{n}^{N}\operatorname{Im}(A_{k,n})$ as Gaussian RVs. 
Assuming that \(A_{k,n}\) are mutually independent across RIS elements and R-UAVs, we have $
\operatorname{Re}(A_k)
\sim
\mathcal{N}\left(
\sum_{n=1}^{N}\mu_{\operatorname{Re}(A_{k,n})},
\sum_{n=1}^{N}\sigma_{\operatorname{Re}(A_{k,n})}^2
\right)$ and $\operatorname{Im}(A_k)
\sim
\mathcal{N}\left(
\sum_{n}^{N}\mu_{\operatorname{Im}(A_{k,n})},
\sum_{n}^{N}\sigma_{\operatorname{Im}(A_{k,n})}^2
\right)$. 
Accordingly, the expectations and variances of $\operatorname{Re}(A)$ and $\operatorname{Im}(A)$ can be expressed as follows:
\begin{equation}\label{mean_varA}
\begin{cases}\!
 \mu_{\operatorname{Re}(A)}\!=\!\sum_{k}^{K}\!\sum_{n}^{N}\mu_{\operatorname{Re}(A_{k,n})},~\mu_{\operatorname{Im}(A)}\!=\!\sum_{k}^{K}\!\sum_{n}^{N}\mu_{\operatorname{Im}(A_{k,n})},\\
\!\sigma_{\operatorname{Re}(A)}^2\!=\!\sum_{k}^{K}\!\sum_{n}^{N}\sigma_{\operatorname{Re}(A_{k,n})}^2,~ \sigma_{\operatorname{Im}(A)}^2\!=\!\sum_{k}^{K}\!\sum_{n}^{N}\sigma_{\operatorname{Im}(A_{k,n})}^2.
\end{cases}  
\end{equation}
Based on the properties of Gaussian RVs, we can calculate the expectations of the second, third, and fourth order moment of $\operatorname{Re}(A)$ and $\operatorname{Im}(A)$, respectively, as follows:
\begin{equation}
\begin{cases}
 \mathbb{E}[\operatorname{Re}(A)^2]=\mu_{\operatorname{Re}(A)}^2+\sigma_{\operatorname{Re}(A)}^2,\\\mathbb{E}[\operatorname{Re}(A)^3]=\mu_{\operatorname{Re}(A)}^3+3\mu_{\operatorname{Re}(A)}\sigma_{\operatorname{Re}(A)}^2,\\
 \mathbb{E}[\operatorname{Re}(A)^4]=\mu_{\operatorname{Re}(A)}^4+6\mu_{\operatorname{Re}(A)}^2\sigma_{\operatorname{Re}(A)}^2+3\sigma_{\operatorname{Re}(A)}^4,
\end{cases}\label{mean_var_nth_order_Re}  
\end{equation}
and
\begin{equation}
\begin{cases}
  \mathbb{E}[\operatorname{Im}(A)]=0,~~~~~~\mathbb{E}[\operatorname{Im}(A)^2]=\sigma_{\operatorname{Im}(A)}^2,\\\mathbb{E}[\operatorname{Im}(A)^3]=0,~~~~~
 \mathbb{E}[\operatorname{Im}(A)^4]=3\sigma_{\operatorname{Im}(A)}^4.
\end{cases}\label{mean_var_nth_order_Im}  
\end{equation}
Therefore, by substituting the above expectation results of the variables, $\mathbb{E}[|U|^2]$ and $\mathbb{E}[|U|^4]$ can be calculated, thus leading to the results of $\mu_{|U|^2}$ and $\sigma_{|U|^2}^2$.

\subsubsection{$1<m_s\leq 2$} The first-order moment \(\mathbb{E}[|U|]\) is required under the circumstances. We first define $R_\zeta=A+E_\zeta$. Since \(\operatorname{Re}(A)\) and \(\operatorname{Im}(A)\) are approximated as independent Gaussian RVs via the CLT and \(E_\zeta\sim\mathcal{CN}(0,\zeta\sigma_E^2)\), the equivalent real-valued vector representation of \(R_\zeta\), denoted by \(\mathbf r_\zeta\), can be expressed as
$\mathbf r_\zeta
\triangleq
\begin{bmatrix}
\operatorname{Re}(R_\zeta)\\
\operatorname{Im}(R_\zeta)
\end{bmatrix}
\sim
\mathcal N\!\left(
\boldsymbol{\mu}_{R_\zeta},
\mathbf C_{R_\zeta}
\right)$,
where the mean vector and covariance matrix are respectively given by 
$\boldsymbol{\mu}_{R_\zeta}
=
\begin{bmatrix}
\mu_{\operatorname{Re}(A)}\\
0
\end{bmatrix}$ and 
$\mathbf C_{R_\zeta}
=
\begin{bmatrix}
\sigma_{x,\zeta}^2 & 0\\
0 & \sigma_{y,\zeta}^2
\end{bmatrix}$,
with
$\sigma_{x,\zeta}^2 \triangleq \sigma_{\operatorname{Re}(A)}^2+\frac{\zeta\sigma_E^2}{2}$ and $
\sigma_{y,\zeta}^2 \triangleq \sigma_{\operatorname{Im}(A)}^2+\frac{\zeta\sigma_E^2}{2}$. 
Therefore, $\mathbb{E}[|U|]$ can be rewritten in a conditional form as follows:
\begin{equation}
\begin{aligned}
&\mathbb{E}[|U|]
=\mathbb E[|D+R_\zeta|] =
\int_{0}^{\infty}
\mathbb E[|U|\,|\,D=d]\,
f_D(d)\,dd \\
\end{aligned}\label{U}
\end{equation}
where
\begin{equation}
	\begin{aligned}
&\mathbb{E}[|U| \mid D=d] \overset{(a)}{=} \int_{-\infty}^{\infty} \int_{-\infty}^{\infty}\frac{\sqrt{x^2 + y^2} }{2\pi \sigma_{x,\zeta} \sigma_{y,\zeta}} \\&\exp\left(-\frac{(x - (d + \mu_{\operatorname{Re}(A)}))^2}{2\sigma_{x,\zeta}^2}\right) \exp\left(-\frac{y^2}{2\sigma_{y,\zeta}^2}\right) dx dy\\
&\overset{(b)}{\approx}
\frac{1}{\pi}
\sum_{i=1}^{N_h}\sum_{j=1}^{N_h}w_i w_j
\sqrt{
\left(d+\mu_{\operatorname{Re}(A)}+\sqrt{2}\sigma_{x,\zeta}x_i\right)^2
+
2\sigma_{y,\zeta}^2 x_j^2
}.
\label{EU_conditional2}
	\end{aligned}
\end{equation}
Since Eq.~\eqref{EU_conditional2} does not admit a simple closed-form expression, here we employ the Gauss-Hermite quadrature to obtain a highly accurate approximation. 
Specifically, for
\(
\int_{-\infty}^{\infty} e^{-t^2} g(t)\,dt
\),
we have
$\int_{-\infty}^{\infty} e^{-t^2} g(t)\,dt
\approx
\sum_{i=1}^{N_h} w_i g(x_i)$, where \(x_i\) and \(w_i\) denote the \(i\)-th Hermite node and weight, respectively, and \(N_h\) is the truncation order. 
Therefore, by applying the variable substitutions $x = \sqrt{2}\sigma_{x,\zeta}u + d + \mu_{\operatorname{Re}(A)}$ and $y = \sqrt{2}\sigma_{y,\zeta}v$, the conditional expectation can be reformulated and tightly approximated as shown in $(b)$,
where \(x_i\) and \(w_i\) denote the \(i\)-th Hermite abscissa and weight, respectively.

Next, we set $C_d = {\widehat{\rho}_{d}\widetilde{\rho}_{d}}/{\sqrt{PL_d}}$ such that $D = C_d |\widetilde{h}_d|$. By introducing the variable substitution $t = {m^d \Omega_s^d {|\widetilde{h}_d|}^2}/({m^d \Omega_s^d {|\widetilde{h}_d|}^2 + m_s^d \Omega_m^d})$, the semi-infinite integration domain $[0, \infty)$ is mapped to a finite interval $[0, 1]$, hence converting \(f_D(d)\,dd\) into the Beta-weighted form as \vspace{-5pt}
\begin{equation}\label{fdd}
f_D(d)\,dd=f_{|\widetilde h_d|}(h)\,dh=
	\frac{t^{m^d-1}(1-t)^{m_s^d-1}}{B(m^d,m_s^d)}\, d t.\vspace{-5pt}
\end{equation} 
For the Gauss-Jacobi quadrature, the integral of a function $f(z)$ weighted by $(1-z)^\alpha(1+z)^\beta$ over the interval $[-1, 1]$ can be approximated as $\int_{-1}^{1} (1-z)^\alpha (1+z)^\beta f(z) dz \approx \sum_{k=1}^{N_J} \omega_k f(z_k)$, where $z_k$ and $\omega_k$ denote the $k$-th Jacobi node and weight, respectively, corresponding to the parameters $\alpha$ and $\beta$, and $N_J$ is the truncation order. 
Therefore, by substituting Eq.~\eqref{fdd} into Eq.~\eqref{U} and applying the linear mapping $z = 2y - 1$, the integration domain is converted to $[-1, 1]$, which exactly matches the standard form of the Gauss-Jacobi quadrature $\int_{-1}^1 f(z)(1-z)^{m_s^d - 1} (1+z)^{m^d - 1} dz$. 
Therefore, 
$\mathbb{E}[|U|]$ can be tightly approximated as follows:\vspace{-5pt}
\begin{equation}
\begin{aligned}
    \mathbb{E}[|U|] &\approx \frac{1}{\pi 2^{m^d+m_s^d-1} B(m^d, m_{s}^d)} \sum_{k=1}^{N_J} \sum_{i=1}^{N_h} \sum_{j=1}^{N_h} \omega_k w_i w_j \\
    &\quad \times \sqrt{\left(\sqrt{2}\sigma_{x,\zeta}x_i + \Xi(z_k) + \mu_{\operatorname{Re}(A)}\right)^2 + 2\sigma_{y,\zeta}^2 x_j^2},
\end{aligned}\label{EU}
\end{equation}
where $\Xi(z_k) = C_d \sqrt{\frac{m_s^d \Omega_m^d}{m^d \Omega_s^d}\frac{z_k+1}{1-z_k}}$ is defined for simplicity. 
Accordingly, \(\mathbb E[|U|]\) is reduced to a finite triple summation, which can be efficiently evaluated in standard numerical software such as MATLAB. Since $\mathbb{E}[|U|^2]$ has been obtained from Eqs.~\eqref{EU2} and \eqref{EX2}, the results of $\mu_{|U|}$ and $\sigma_{|U|}^2$ can be obtained, which completes the proof.

\vspace{-15pt}
\section{Derivation of the CDF of $\gamma_\zeta$}\label{AppendixB}\vspace{-10pt}
For $m_s>2$, we start from the complementary cumulative distribution function (CCDF) of $\gamma_\zeta$, which is denoted by $\widetilde{F}_{\gamma_\zeta}(\gamma)$ and given by \vspace{-5pt}
\begin{equation}\label{CDFgammazeta_1}
	\begin{aligned}
		\widetilde{F}_{\gamma_\zeta}(\gamma)& = 1-{F}_{\gamma_\zeta}(\gamma)=\text{Pr}\left[|U|^2 > \left(|\varepsilon_\zeta|^2+\frac{1}{\bar{\gamma}}\right)\gamma\right]\\
		&= \int_{0}^\infty \widetilde{F}_{|U|^2}\left(\gamma \varepsilon+\frac{\gamma}{\bar{\gamma}}\right)f_{\scriptscriptstyle |\varepsilon_\zeta|^2}(\varepsilon) d\varepsilon,\vspace{-10pt}
	\end{aligned}
\end{equation}
where $\widetilde{F}_{|U|^2}(u)$ denotes the CCDF of $|U|^2$. According to \cite[Eq. (8.356.3)]{2007Table}, we have  $\widetilde{F}_{|U|^2}(u)={\Gamma(k_{u_2},{u}/{\theta_{u_2}})}/{\Gamma(k_{u_2})}$, where $\Gamma(\alpha, x) = \int_{x}^{\infty} \exp(-t) t^{\alpha-1} \, dt$ is the upper incomplete gamma function\cite[Eq. (8.350.2)]{2007Table}.  
Then substituting Eq.~\eqref{PDF_varepsilon} into Eq.~\eqref{CDFgammazeta_1}, the CCDF can be rewritten as follows:\vspace{-5pt}
\begin{equation}\label{CDFgammazeta_2}
	\begin{aligned}
		\widetilde{F}_{\gamma_\zeta}(\gamma)& =\frac{1}{\Gamma(k_{u_2})\sigma_{\varepsilon_\zeta}^2} \times\\&\int_{0}^\infty \left(\int_{\frac{\gamma \varepsilon}{\theta_{u_2}}+\frac{\gamma}{\bar{\gamma}\theta_{u_2}}}^\infty \exp(-t)t^{k_{u_2}-1}dt\right) 
		\exp\left(-\frac{\varepsilon}{\sigma_{\varepsilon_\zeta}^2}\right)d\varepsilon.\vspace{-5pt}
	\end{aligned}
\end{equation}
After exchanging the integration order, the original integral range becomes $t\ge {\gamma}/{(\bar{\gamma}\theta_{u_2})}$ with $0\le \varepsilon \le{\theta_{u_2} t}/{\gamma}-{1}/{\bar{\gamma}}$. The resulting double integral $I$ can then be calculated as follows:\vspace{-5pt}
\begin{equation}\label{CDFgammazeta_3}
	\begin{aligned}
		I &\overset{(a)}{=} \int_\frac{\gamma }{\bar{\gamma}\theta_{u_2}}^{\infty}\left[\int_{0}^{\frac{\theta_{u_2} t}{\gamma}-\frac{1}{\bar{\gamma}}} \exp\left(-\frac{\varepsilon}{\sigma_{\varepsilon_\zeta}^2}\right) d\varepsilon \right] \exp(-t)t^{k_{u_2}-1} dt\\
		&\overset{(b)}{=}\sigma_{\varepsilon_\zeta}^2 \left[\int_\frac{\gamma }{\bar{\gamma}\theta_{u_2}}^{\infty} \exp(-t)t^{k_{u_2}-1}\right.\\
		&\left.~~~~~~~~~~- \exp\left(\frac{1}{\bar{\gamma}\sigma_{\varepsilon_\zeta}^2}\right) \exp\left(-t-\frac{\theta_{u_2} t}{\sigma_{\varepsilon_\zeta}^2 \gamma} \right)t^{k_{u_2}-1} dt\right]\\
		&\overset{(c)}{=}\sigma_{\varepsilon_\zeta}^2\left[\Gamma(k_{u_2},\frac{\gamma}{\bar{\gamma}\theta_{u_2}}) \right.\\
		&\left.- \exp\left(\frac{1}{\bar{\gamma}\sigma_{\varepsilon_\zeta}^2}\right)\left(\frac{\sigma_{\varepsilon_\zeta}^2\gamma}{\theta_{u_2}+\sigma_{\varepsilon_\zeta}^2 \gamma}\right)^{k_{u_2}}\Gamma(k_{u_2},\frac{\theta_{u_2}+\sigma_{\varepsilon_\zeta}^2\gamma }{\bar{\gamma}\theta_{u_2}\sigma_{\varepsilon_\zeta}^2})\right], 
	\end{aligned}
\end{equation}
where $(b)$ follows from directly evaluating the inner integral in $(a)$. The first term of $(c)$ can be calculated using \cite[Eq. (8.350.2)]{2007Table}, while the second term can be derived by applying variable substitution as well as \cite[Eq. (8.350.2)]{2007Table}. Finally, substituting Eq.~\eqref{CDFgammazeta_3} into Eq.~ \eqref{CDFgammazeta_2}, and then combining with Eq.~\eqref{CDFgammazeta_1},  the closed-form expression of the CDF of $\gamma_\zeta$ is obtained.

For $1<m_s \leq 2$, ${F}_{\gamma_\zeta}(\gamma)$ can be expressed as follows:
\begin{equation}
	\begin{aligned}
		&{F}_{\gamma_\zeta}(\gamma) =\text{Pr}\left[|U| < \sqrt{\left(|\varepsilon_\zeta|^2+\frac{1}{\bar{\gamma}}\right)\gamma}\right]\\
		&= \int_{0}^\infty {F}_{|U|}\left(\sqrt{\gamma \varepsilon+\frac{\gamma}{\bar{\gamma}}}\right)f_{\scriptscriptstyle |\varepsilon_\zeta|^2}(\varepsilon) d\varepsilon\\
		&\overset{(a)}{=}\frac{1}{\Gamma(k_{u_1})} \int_{0}^\infty \gamma\left(k_{u_1}, \frac{\sqrt{\gamma ( \sigma_{\varepsilon_\zeta}^2 z +\frac{1}{\bar{\gamma}})}}{\theta_{u_1}}\right)\exp\left(-z\right)dz \\
		&\overset{(b)}{\approx}\frac{1}{\Gamma(k_{u_1})} \sum_{i=1}^{N_{q}} w_{i} \gamma\left(k_{u_1}, \frac{\sqrt{\gamma ( \sigma_{\varepsilon_\zeta}^2 y_i +\frac{1}{\bar{\gamma}})}}{\theta_{u_1}}\right),
	\end{aligned}
\end{equation}
where $(a)$ can be derived with the variable change $z=\varepsilon/{\sigma_{\varepsilon_\zeta}^2}$. To overcome this intractability and provide an efficient analytical evaluation, we employ the Gauss-Laguerre quadrature method \cite{Gauss-Laguerre}. Specifically, an integral of the form $\int_{0}^{\infty} e^{-y} f(y) dy$ can be accurately approximated as $\sum_{i=1}^{N_q} w_i f(y_i)$, where $N_q$ denotes the truncation order of the quadrature, $y_i$ is the $i$-th root of the Laguerre polynomial $L_{n}(y)$\cite[Eq.~(8.970.1)]{2007Table}, 
and $w_i$ represents the corresponding weight factor given by $w_i = \frac{y_i}{(N_q + 1)^2 [L_{N_q+1}(y_i)]^2}$. 
Then, $(b)$ can be obtained by applying the Gauss-Laguerre quadrature approximation, which admits a highly tight finite-sum approximation.

\vspace{-10pt}
\section{Derivation of the Asymptotic EC}\label{AppendixC}
With high SNR, the asymptotic expression of $\gamma_\zeta$ for $\zeta$-Model can be obtained by utilizing the approximation $1/\bar{\gamma}\rightarrow 0$, yielding ${\gamma}_\zeta^{\text{asy}}={|U|^2}/{|\varepsilon_\zeta|^2}$. Since $|\varepsilon_\zeta|^2$ follows an exponential distribution, it can be regarded as a special case of the Gamma distribution, i.e., $|\varepsilon_\zeta|^2 \sim \mathrm{Gamma}(1, \sigma_{\varepsilon_\zeta}^2)$. 
For the regime of $m_s > 2$, since we have $|U|^2 \sim \mathrm{Gamma} (k_{u_2}, \theta_{u_2})$, ${\gamma}_\zeta^{\text{asy}}/c\sim\beta^{\prime}(k_{u_2},1)$, where $\beta^{\prime}(\cdot)$ represents the Beta prime distribution\cite{Betaprime} and $c={\theta_{u_2}}/{\sigma_{\varepsilon_\zeta}^2}$ can be obtained with the properties of the Gamma distribution. By exploiting the deterministic transformation between the Beta and Beta-prime distributions,  ${\gamma}_\zeta^{\text{asy}}$ can be rewritten as ${\gamma}_\zeta^{\text{asy}}= {c\xi}/({1-\xi})$, where the RV $\xi$ follows a Beta distribution with parameters $(k_{u_2}, 1)$, denoted by $\xi\sim \mathrm{Beta}(k_{u_2},1)$. By substituting ${\gamma}_\zeta^{\text{asy}}$ into Eq.~\eqref{EC}, the expectation term $\mathcal{E}_{\gamma_{\zeta}^{\text{asy}}}=\mathbb{E}_{\gamma_\zeta^{\text{asy}}}\left[(1+{\gamma_\zeta^{\text{asy}}})^{-\beta} \right]$ 
can be expressed as follows:
\begin{equation}\label{EC2ub1}  
\begin{aligned}
\mathcal{E}_{\gamma_{\zeta}^{\text{asy}}} &=\mathbb{E}\left[\left(1+ \frac{c\xi}{1-\xi}\right)^{-\beta}\right]=\mathbb{E}\left[\left(\frac{1-(1-c)\xi}{1-\xi}\right)^{-\beta}\right].
\end{aligned}
\end{equation}
Substituting the PDF of $\xi$ into Eq.~\eqref{EC2ub1}, we obtain 
\begin{equation}\label{EC2ub2}  
\begin{aligned}
\mathcal{E}_{\gamma_{\zeta}^{\text{asy}}}&=\int_{0}^1 \frac{(1-\xi)^\beta}{[1-(1-c)\xi]^\beta}k_{u_2} \xi^{k_{u_2}-1}d\xi.
\end{aligned}
\end{equation}
According to \cite[Eq.~(9.111)]{2007Table}, Eq.~\eqref{EC2ub2} can be directly calculated as 
\begin{equation}\label{EC2ub3}  
\begin{aligned}
\mathcal{E}_{\gamma_{\zeta}^{\text{asy}}}&=k_{u_2} B(k_{u_2},\beta+1){}_2 F_1\left(\beta,k_{u_2}; k_{u_2}+\beta+1; 1-c\right),
\end{aligned}
\end{equation}
where ${}_2 F_1\left(\alpha, \beta; \gamma; -z\right)$ denotes the hypergeometric function. 

For the regime of $1<m_s \leq 2$, we have $|U| \sim \mathrm{Gamma} (k_{u_1}, \theta_{u_1})$. After substituting $f_{|U|}(u)$ and $f_{|\varepsilon_\zeta|^2}(\varepsilon)$ and applying the change of variables $x = u^2$, the asymptotic EC expectation term can be formulated as follows:
\begin{equation}\label{EC2ub_msleq2_integral}
    \mathcal{E}_{\gamma_{\zeta}^{\text{asy}}} = \int_0^\infty \frac{x^{\frac{k_{u_1}}{2}-1} e^{-\frac{\sqrt{x}}{\theta_{u_1}}}}{2 \theta_{u_1}^{k_{u_1}} \Gamma(k_{u_1})} \left(\int_0^\infty \left(1+\frac{x}{\varepsilon}\right)^{-\beta} \frac{e^{-\frac{\varepsilon}{\sigma_{\varepsilon_\zeta}^2}}}{\sigma_{\varepsilon_\zeta}^2} d\varepsilon \right) dx .
\end{equation}
Similar to the derivation process for $\mathcal{E}_{\gamma_{0}^s}$, we express the binomial term $(1+x/\varepsilon)^{-\beta}$ using the Mellin-Barnes contour integral representation \cite[Eq.~(9.113)]{2007Table} and then exchange the order of integration, which yields 
\begin{equation}
	\begin{aligned}
	&\mathcal{E}_{\gamma_{\zeta}^{\text{asy}}} = \frac{1}{\Gamma(k_{u_1})\Gamma(\beta)}\\ &\times \frac{1}{2\pi i}
	\int_{\mathcal L} \Gamma(s)\Gamma(\beta-s) \Gamma(s+1) \Gamma(k_{u_1}-2s) \left(\frac{\sigma_{\varepsilon_\zeta}^2}{\theta_{u_1}^2}\right)^s ds.
	\end{aligned}
\end{equation}
Then, by applying the Legendre duplication formula\cite[Eq. (8.335.1)]{2007Table} to $\Gamma(k_{u_1}-2s)$ and recasting the resulting contour integral into the standard Meijer's $G$-function form\cite[Eq.~(9.301)]{2007Table}, $\mathcal{E}_{\gamma_{\zeta}^{\text{asy}}}$ is finally obtained as follows:
\begin{equation}\label{EC2ub_msleq2_final}
\mathcal{E}_{\gamma_{\zeta}^{\text{asy}}}
= \frac{2^{k_{u_1}-1}}{\sqrt{\pi}\Gamma(k_{u_1})\Gamma(\beta)}
G_{3,2}^{2,3}\!\left(
\frac{4\theta_{u_1}^2}{\sigma_{\varepsilon_\zeta}^2}
\;\middle|\;
\begin{matrix}
1-\beta,1-\frac{k_{u_1}}{2},\frac{1-k_{u_1}}{2}\\
0, 1
\end{matrix}
\right).
\end{equation}
Finally, by substituting Eqs.~\eqref{EC2ub3} and \eqref{EC2ub_msleq2_final} into Eq.~\eqref{EC}, Eq.~\eqref{EC2ub} can then be obtained, which completes the proof.
\end{appendices}
\bibliography{References}
\bibliographystyle{IEEEtran}
\end{document}